\documentclass[11pt,reqno]{article}

\usepackage{amsmath, amssymb, amsfonts, amsbsy, amsthm,
latexsym,graphicx,subfigure}
\usepackage{cancel}
\usepackage{pifont}
\usepackage{pdfpages}
\usepackage{epstopdf}
\advance \topmargin by -\headheight
\advance \topmargin by -\headsep
\evensidemargin \oddsidemargin
\newtheorem{theorem}{Theorem}[section]
\newtheorem{lemma}[theorem]{Lemma}

\theoremstyle{definition}

\newtheorem{design-crit}[theorem]{Design Criterion}
\newtheorem{remark}[theorem]{Remark}

\theoremstyle{remark}

\newcommand{\E}{\mathbb{E}}
\newcommand{\Prob}{\mathbb{P}}

\newcommand{\R}{{\mathbb{R}}}

\newcommand{\ra}{{\rangle}}
\newcommand{\la}{{\langle}}

\newcommand{\ov}{\overline}

\newcommand{\vare}{{\varepsilon}}
\newcommand{\eps}{{\epsilon}}

\newcommand{\rank}{\mathrm{rank}}
\newcommand{\card}{\mathrm{card}}
\newcommand{\sgn}{\mathrm{sgn}}
\newcommand{\supp}{\mathrm{supp}}
\newcommand{\eff}{\mathrm{eff}}
\newcommand{\tr}{{\mathrm{Tr}}}

\newcommand{\bX}{{\mathbf X}}
\newcommand{\bY}{{\mathbf Y}}
\newcommand{\bA}{{\mathbf A}}
\newcommand{\bE}{{\mathbf E}}
\newcommand{\bx}{{\mathbf x}}
\newcommand{\bv}{{\mathbf v}}
\newcommand{\bw}{{\mathbf w}}
\newcommand{\bu}{{\mathbf u}}
\newcommand{\by}{{\mathbf y}}
\newcommand{\bz}{{\mathbf z}}
\newcommand{\bZ}{{\mathbf Z}}
\newcommand{\bh}{{\mathbf h}}
\newcommand{\be}{{\mathbf e}}
\newcommand{\bB}{{\mathbf B}}

\newcommand{\bC}{{\mathbf C}}
\newcommand{\bP}{{\mathbf P}}

\newcommand{\AP}{{\bf A_P}}
\newcommand{\tAP}{{\bf \tilde{A}_P}}

\makeatletter

\newcommand{\Rmnum}[1]{\expandafter\@slowromancap\romannumeral #1@}
\makeatother

\makeatletter
\newcommand{\specificthanks}[1]{\@fnsymbol{#1}}
\makeatother

\begin{document}

\title{Sparse Recovery of Fusion Frame Structured Signals
}



\author{Ula\c{s} Ayaz\thanks{Laboratory for Information \& Decision Systems, Massachusetts 
  Institute of Technology, Cambridge, MA, USA.}}





\maketitle



\begin{abstract}

Fusion frames are collection of subspaces which provide a redundant
representation of signal spaces. They generalize classical frames by
replacing frame vectors with frame subspaces. This paper considers
the sparse recovery of a signal from a fusion frame. We use a
block sparsity model for fusion frames and then show that sparse
signals under this model can be compressively sampled and
reconstructed in ways similar to standard Compressed Sensing (CS).
In particular we invoke a mixed $\ell_1/\ell_2$ norm minimization in
order to reconstruct sparse signals. In our work, we show that
assuming a certain incoherence property of the subspaces and the
apriori knowledge of it allows us to improve recovery when compared
to the usual block sparsity case.
\end{abstract}

\section{Introduction}
The problem of recovering sparse signals in $\R^N$ from $m < N$
measurements has drawn a lot of attention in recent years
\cite{carota06,cata06,do06-2}. Compressed Sensing (CS) achieves such
performance by imposing a sparsity model on the signal of interest.
The sparsity model, combined with randomized linear acquisition,
guarantees that certain non-linear reconstruction can be used to
efficiently and accurately recover the signal.

Classical frames are nowadays a standard notion for modeling
applications, in which redundancy plays a vital role such as filter
bank theory, sigma-delta quantization, signal and image processing
and wireless communications. Fusion frames are a
relatively new concept which is potentially useful when a single frame
system is not sufficient to represent the whole sensing mechanism
efficiently. Fusion frames were first introduced in \cite{Casazza04}
under the name of `frames of subspaces' (see also the survey
\cite{Casazza13}). They analyze signals by projecting them onto
multidimensional subspaces, in contrast to frames which consider
only one-dimensional projections. In the literature, there have been
a number of applications where fusion frames have proven to be
useful practically, such as \textit{distributed processing}
\cite{Kutyniok09, Casazza07}, \textit{parallel processing}
\cite{Bjorstad91,Oswald97}, \textit{packet encoding}
\cite{Bodmann07}.

In this paper, we consider the recovery of sparse signals from a
fusion frame. Signals of interest are collections of vectors from
fusion frame subspaces which can be considered as a `block' signal.
In other words, say we have $N$ subspaces, then we have a collection
of $N$ vectors which is the (block) signal we wish to recover. In
addition to block structure, the signals from a fusion frame have
the property that each block belongs to a
particular fusion frame subspace. We
are then allowed to observe linear measurements of those vectors and 
we aim to reconstruct the original signal from those measurements. In
order to do so, we use ideas from CS. We assume a `block' sparsity
model on the signals to be recovered where a few of the vectors in
the collection are assumed to be nonzero. For instance, we are not
interested whether each vector is sparse or not within the subspace
it belongs to. For the reconstruction, a mixed $\ell_1/\ell_2$
minimization is invoked in order to capture the structure that comes
with the sparsity model.

This problem was studied before in \cite{Boufounos09}
where the authors proved that it is sufficient for recovery to take
$m \gtrsim s \ln(N)$ random measurements. Here $s$ is the sparsity
level and $N$ is the number of subspaces. It is worth emphasizing
that our model assumes that the signals of interest lie in
particular subspaces which is not assumed in block sparsity
problems. In this paper, our goal is to improve the results in
\cite{Boufounos09} when the subspaces have a certain incoherence
property , i.e., they have nontrivial mutual angles between them,
and we assume to know them. Recently the authors in \cite{AyazU13}
have studied this problem in the uniform recovery setting. Our focus
in this paper will be on the nonuniform recovery of signals from a
fusion frame. To our best knowledge, a result of this
nature is not available in the literature.

\subsection{Notation} 
We denote Euclidean norm of vectors by $\|\cdot\|_2$ and the spectral norm of matrices by $\| \cdot \|$.
Boldface notation refers to block vectors
and matrices throughout this paper. Let $[N]=\{1,2,\ldots,N\}$. For a block matrix $\bA =
(a_{ij} B_{ij})_{i \in [m], j \in [N] } \in \R^{m d \times N d}$ with
blocks $B_{ij} \in \R^{d \times d}$, we denote the $\ell$-th
block column by $\bA_\ell = (a_{i \ell} B_{i \ell})_{i \in [m]} \in \R^{m d \times d}$ 
and the column submatrix restricted to $S \subset
[N]$ by $\bA_S = (a_{ij} B_{ij})_{i \in [m],j \in S}$. Similarly for a
block vector $\bx=(x_i)_{i=1}^N \in \R^{N d}$ with $x_i \in \R^d$, we denote
the vector $\bx_S=(x_i)_{i \in S}$ restricted to $S$. The complement $[N] \setminus S$
is denoted $\ov{S}$. The $\ell_\infty \to \ell_\infty$ and $\ell_2 \to \ell_\infty$ norms of a
matrix $A \in \R^{m \times n}$ are given as
$$
\|A\|_\infty = \max_{ i \in [m]} \sum_{j=1}^n A_{ij}, \ \ \text{ and } \ \ \|A\|_{2,\infty} = \max_{ i \in [m]} \left(\sum_{j=1}^n A_{ij}^2 \right)^{1/2},
$$
respectively.
Furthermore, for a block vector $\bx = (x_i)_{i=1}^N$ with $x_i \in \R^d$ we define the
$\ell_{2,\infty}$-norm as
$$
\|\bx\|_{2,\infty}= \max_{i \in [N]} \|x_i\|_2.
$$
Given a block vector $\bz$ with $N$ blocks from $\R^d$, we define
$\sgn(\bz) \in \R^{dN}$ analogously to the scalar case as
$$ \sgn(\bz)_i = \left\{ \begin{array}{ccc} \frac{z_i}{\|z_i\|_2} & : &
\|z_i\|_2 \neq 0, \\ 0 & : & \|z_i\|_2 =0.
\end{array} \right. $$

\section{Background on Fusion Frames}
A {\em fusion frame} for $\R^d$ is a collection of $N$ subspaces
$W_j \subset \R^d$ and associated weights $v_j$ that satisfy
\begin{align}\label{FF_property}
A \|x\|_2^2 \leq \sum_{j=1}^N v_j^2\|P_j x\|_2^2 \leq B\|x\|_2^2
\end{align}
for all $x \in \R^d$ and for some universal fusion frame bounds $0 <
A\leq B < \infty$, where $P_j \in \R^{d \times d}$ denotes the
orthogonal projection onto the subspace $W_j$. For simplicity we
assume that the dimensions of the $W_j$ are equal, that is $\text{dim}(W_j)=
k$. For a fusion frame $(W_j)_{j=1}^N$, let us define the space
$$\mathcal{H}=\{ (x_j)_{j=1}^N : x_j \in W_j, \ \forall j \in [N]\} \subset
\R^{d \times N},
$$
where we denote $[N]=\{1,\ldots,N\}$. The {\em mixed
$\ell_{2,1}$-norm} of a vector $\bx \equiv (x_j)_{j=1}^N \in \mathcal{H}$ is defined as
\begin{align*}
\|\bx\|_{2,1} = \sum_{j=1}^N \|x_j\|_2.
\end{align*}
Furthermore, the `$\ell_0$- block norm' of $\bx \in \mathcal{H}$ (which
is actually not even a quasi-norm) is defined as
$$\|\bx\|_0 = \sharp \{j \in [N] : x_j \neq 0\}.$$
We say that a vector $\bx \in \mathcal{H}$ is $s$-sparse, if $\|\bx\|_0 \leq
s$.

In this paper we consider the recovery of sparse vectors from the set $\mathcal{H}$ which
collects all vectors from a fusion frame subspaces. However our results do not
assume necessarily that subspaces satisfy the fusion frame property \eqref{FF_property}
but rather assume they satisfy an incoherence
property as explained in Section~\ref{incoher}.

\subsection{Relation with other recovery problems}
A special case of the sparse recovery problem above appears when all
subspaces coincide with the ambient space $W_j=\R^d$ for all $j$.
Then the problem reduces to the well studied {\em joint sparsity
setup} \cite{fora08} in which all the vectors have the same sparsity
structure. In order to see this, say we have $N$ vectors $\{x_1,
\ldots, x_N\}$ in $\R^d$ and write them as rows of a matrix $X \in
\R^{N \times d}$. Assuming only $s$ of the rows are non-zero, the $d$
vectors consisting of the columns of the matrix $X$ have the same
common support set, i.e., are \textit{jointly} sparse. To be more
precise, if $(X_i)_{i=1}^d$ denote the columns of $X$, $\supp(X_i) =
\supp(X_j)$ for all $i \neq j$. 

Furthermore, our problem is itself a special case of {\em block
sparsity setup} \cite{Eldar08}, with significant additional
structure that allows us to enhance existing results. Without the
subspace structure, we would have $N$ vectors $\{x_1, \ldots, x_N\}$ 
in $\R^d$ where only $s$ of them are non-zero, i.e., $\card\{ i \in
[N]: x_i \neq 0\} =s$. The reason this is called block sparsity
is because we count the vectors which are non-zero as a
\textit{block} rather than checking if its entries are sparse. 

Another related concept is \textit{group sparsity} \cite{Rao12}
where each entry of the vector is assigned to a predefined group and
sparsity counts the groups that are active in the support set of the
vector. In mathematical notation, consider a vector $x$ of length
$p$ and assume that its coefficients are grouped into sets $\{G_i\}_{i=1}^N$,
such that for all $i \in [N]$, $G_i \subset [p]$. The vector $x$ is
group $s$-sparse if the non-zero coefficients lie in $s$ of the groups
$\{G_i\}_{i=1}^N$. Formally,
$$
\card \{i \in [N]: \supp(x) \cap G_i \neq 0\}=s.
$$
This is similar to the block sparsity with $p=Nd$ except that the groups
may be allowed to overlap, i.e., $G_i \cap G_j \neq \emptyset$. We
also note that, for sparse recovery, a natural proxy for group sparsity becomes 
the norm $\sum_i \|x_{G_i}\|_2$. 

Finally in the case $d=1$, the projections equal 1, and hence the
problem reduces to the {\em standard CS recovery problem} $Ax=y$
with $x \in \R^N$ and $y \in \R^m$.

\subsection{Sparse recovery problem}
Our measurement model is as follows. Let $\bx^0=(x^0_j)_{j=1}^N \in \mathcal{H}$ be an $s$-sparse
and we observe $m$ linear combinations of those vectors
$$
\by=(y_i)_{i=1}^m = \left(\sum_{j=1}^N a_{ij} x_j^0 \right)_{i=1}^m,
\ \ y_i \in \R^d,
$$
for some scalars $(a_{ij})$.
Let us denote the block matrices ${\bf A_I}=(a_{ij} I_d)_{i \in [m],
j\in[N]}$ and $\AP = (a_{ij}P_j)_{i \in [m],j \in [N]}$ that consist
of the blocks $a_{ij}I_d$ and $a_{ij}P_j$ respectively. Here $I_d$
is the identity matrix of size $d \times d$. Then we can formulate
this measurement scheme as
$$ \by = {\bf A_I} \bx^0 = \AP \bx^0, \ \text{ for } \bx^0 \in \mathcal{H}.$$
We replaced ${\bf A_I}$ by $\AP$ since the relation $P_j x_j=x_j$
holds for all $j \in [N]$. We wish to recover $\bx^0$ from those
measurements. This problem can be formulated as the following
optimization program
\begin{align*}
(L0) \ \ \hat{\bx}= \text{argmin}_{\bx \in \mathcal{H}} \|\bx\|_0 \ \
s.t. \ \ \AP\bx=\by.
\end{align*}
The optimization program ($L0$) is NP-hard. Instead we propose the
following convex program the former
\begin{align*}
(L1) \ \ \hat{\bx}= \text{argmin}_{\bx \in \mathcal{H}} \|\bx\|_{2,1} \
\ s.t. \ \ \AP \bx=\by.
\end{align*}
We shall show that block sparse signals can be accurately recovered
by solving ($L1$). In the rest of the paper $\bP$ denotes the $N
\times N$ block diagonal matrix where the diagonals are projection
matrices $P_j$, $j \in [N]$. $\bP_S$ denotes the $s \times s$ block
diagonal matrix with entries $P_i$, $i \in S$.

\subsection{Incoherence matrix }\label{incoher}
In this section we define an incoherence matrix $\Lambda$ associated
with a fusion frame
\begin{align}\label{definco}
\Lambda=(\alpha_{ij})_{i,j \in [N]},
\end{align}
where $\alpha_{ij}=\|P_iP_j\|$ for $i \neq j \in [N]$ and
$\alpha_{ii}=0$. Note that $\|P_iP_j\|$ equals the largest absolute
value of the cosines of the principle angles between $W_i$ and
$W_j$. Then, for a given support set $S \subset [N]$, we denote $\Lambda_S$ as the submatrix of $\Lambda$
with columns restricted to $S$, and $\Lambda^S$ as the the submatrix
with columns and rows restricted to $S$.
Then we note the following norms
$$
\|\Lambda_S\|_{2,\infty} = \max_{ i \in [N]} \left( \sum_{j \in S, j \neq i} \|P_i P_j\|^2 \right)^{1/2} \ \ \text{ and } \ \ \|\Lambda^S\|_{2,\infty} = \max_{ i \in S}
\left( \sum_{j \in S, j \neq i} \|P_i P_j\|^2 \right)^{1/2},
$$
and
$$
\|\Lambda_S\|_\infty = \max_{ i \in [N]} \sum_{j \in S, j \neq i} \|P_i P_j\| \ \ \text{ and } \ \ \|\Lambda^S\|_\infty = \max_{ i \in S} \sum_{j \in S, j \neq i} \|P_i P_j\|.
$$
Moreover, we define the parameter $\lambda$ as
\begin{align*}
\lambda = \max_{i \neq j} \|P_iP_j\|, \ \ i,j \in [N].
\end{align*}
Clearly $\lambda$ equals the largest entry of $\Lambda$. In addition
it holds that
\begin{align*}
\| \Lambda_S\|_{2,\infty} \leq \| \Lambda_S\|_\infty \leq \lambda s
\end{align*}
for any $S$ with $|S|=s$. If the subspaces are all orthogonal to each other, i.e.
$\lambda=0$ and $\Lambda=0$, then only one measurement $y= \sum_i
a_i x_i$ suffices to recover $\bx^0$. In fact, since $x_i \perp
x_j$, we have
$$x_i = \frac{1}{a_i} P_i y.$$
This observation suggests that fewer measurements are necessary when
$\lambda$ gets smaller. In this work our goal is to provide a solid
theoretical understanding of this observation.

\section{Nonuniform recovery with Bernoulli matrices}

In this section we study nonuniform recovery from fusion frame
measurements. Our main result states that a fixed sparse signal can be
recovered with high probability using a random draw of a Bernoulli
random matrix $A \in \R^{m \times N}$ whose entries take value
$\pm 1$ with equal probability. Also we assume that fusion frames
$(W_j)_{j=1}^N$ for $\R^d$ have $\dim(W_j)=k$ for all $j$.

\begin{theorem}\label{nonuni_main}Let $\bx \in \mathcal{H}$ be $s$-sparse and $S=\supp(\bx)$. Let $A \in \R^{m \times N}$ be
Bernoulli matrix and a fusion frame $(W_j)_{j=1}^N$ be given with
the incoherence matrix $\Lambda$. Assume that
\begin{align}\label{nonuni_cond}
m \geq C (1 + \|\Lambda_S\|_\infty ) \ln(N) \ln(sk) \ln(\vare^{-1}),
\end{align}
where $C > 0$ is a universal constant. Then with probability at
least $1-\vare$, $(L1)$ recovers $\bx$ from $\by=\AP \bx$.
\end{theorem}

If the subspaces are not
equi-dimensional, one can replace $sk$ term in Condition~\eqref{nonuni_cond} by
$\sum_{i\in S} \dim(W_j)$, where $\dim(W_j)=k_j$. We prove Theorem~\ref{nonuni_main}
in Section \ref{golf_proof}. The proof
relies on the recovery condition of Lemma~\ref{inexact} via an
inexact dual. 

\paragraph{Gaussian Case.}
We also state a similar result for Gaussian random matrices without a
proof. These matrices have entries as independent standard
Gaussian random variables, i.e., $A_{ij}=g_{ij} \sim \mathcal{N}(0,1)$. 

\begin{theorem}\label{gauss_main}Let $\bx \in \mathcal{H}$ be $s$-sparse. Let $A \in \R^{m \times N}$ be a Gaussian matrix
and $(W_j)_{j=1}^N$ be given with parameter $\lambda \in [0,1]$ and
$\dim(W_j)=k$ for all $j$. Assume that
\begin{align*}
m \geq \tilde{C} (1+\lambda s) \ln^2 (6 Nk ) \ln^2(\vare^{-1}),
\end{align*}
where $\tilde{C} > 0$ is a universal constant. Then with probability
at least $1-\vare$, $(L1)$ recovers $\bx$ from $\by=\AP \bx$.
\end{theorem}

The proof of this result can be found in \cite[Section~3.3]{Ayaz14}. It
also follows the inexact dual lemma 
however proceeds in a rather different way the
Bernoulli case since the
probabilistic tools we use for proving Theorem \ref{nonuni_main} apply only for
bounded random variables.  Therefore we apply other tools which make
proofs considerably long and so we do not present them
here.

\subsection{Recovery lemma}
This section gives a sufficient condition for recovery of fixed
sparse vectors based on an ``inexact dual vector". Sufficient
conditions involving an exact dual vector were given in
\cite{Fuchs04,Tropp05}. The modified inexact version is due to
Gross \cite{Gross09}. Below, $A_{|\mathcal{H}}$ restricts the matrix $A$ to its range $\mathcal{H}$.

\begin{lemma}\label{inexact}
Let $A \in \R^{m \times N}$ and $(W_j)_{j=1}^N$ be a fusion frame
for $\R^d$ and $x \in \mathcal{H}$ with support $S$. Assume that
\begin{align}\label{dual_cond1}
\|[(\AP)_S^* (\AP)_S]_{|\mathcal{H}}^{-1} \| \leq 2 \ \ \text{ and } \
\ \max_{\ell \in \ov{S}} \|(\AP)_S^* (\AP)_\ell \| \leq 1.
\end{align}
Suppose there exists a block vector $\bu \in \R^{Nd}$ of the form
$\bu = \AP^* \bh$ with block vector $\bh \in \R^{md}$ such that
\begin{align}\label{dual_cond2}
\|\bu_S- \sgn(\bx_S)\|_2 \leq 1/4 \ \ \text{ and } \ \ \max_{i \in
\ov{S}} \|u_i\|_2 \leq 1/4.
\end{align}
Then $x$ is the unique minimizer of $\|\bz\|_{2,1}$ subject to $\AP
\bz=\AP \bx$.
\end{lemma}

\begin{proof}
The proof follows \cite[Theorem 4.32]{Foucart13} and generalizes it
to the block vector case. For convenience we give the details here. Let
$\hat{\bx}$ be a minimizer of $\|\bz\|_{2,1}$ subject to ${\bf
A_P}\bz=\AP \bx$. Then $\bv=\hat{\bx}-\bx \in \mathcal{H}$ satisfies
$\AP \bv= {\bf 0}$. We need to show that $\bv={\bf 0}$. First we
observe that
\begin{align}\label{norm1}
\| \hat{\bx} \|_{2,1} &= \|\bx_S + \bv_S\|_{2,1} +
\|\bv_{\ov{S}}\|_{2,1} =
\la \sgn (\bx_S + \bv_S), (\bx_S + \bv_S) \ra + \|\bv_{\ov{S}}\|_{2,1} \notag \\
& \geq \la \sgn (\bx_S), (\bx_S+\bv_S) \ra + \|\bv_{\ov{S}}\|_{2,1} \notag \\
&= \|\bx_S\|_{2,1} + \la \sgn (\bx_S), \bv_S \ra +
\|\bv_{\ov{S}}\|_{2,1}.
\end{align}
For $\bu= \AP^* \bh$ it holds
$$ \la \bu_S,\bv_S \ra = \la \bu,\bv \ra - \la \bu_{\ov{S}} , \bv_{\ov{S}} \ra
= \la \bh, \AP \bv \ra - \la \bu_{\ov{S}} , \bv_{\ov{S}} \ra = -\la
\bu_{\ov{S}} , \bv_{\ov{S}} \ra.$$ Hence,
\begin{align*}
\la \sgn (\bx_S), \bv_S \ra &= \la \sgn (\bx_S)-\bu_S, \bv_S \ra +
\la
\bu_S,\bv_S \ra \\
&= \la \sgn (\bx_S)-\bu_S, \bv_S \ra - \la \bu_{\ov{S}},\bv_{\ov{S}}
\ra.
\end{align*}
The Cauchy-Schwarz inequality together with \eqref{dual_cond2}
yields \begin{align*} |\la \sgn (\bx_S), \bv_S \ra | \leq \|\sgn
(\bx_S)-\bu_S\|_2 \|\bv_S\|_2 + \max_{i \in \ov{S}} \|u_i\|_2
\|\bv_{\ov{S}}\|_{2,1} \leq \frac{1}{4} \|\bv_S\|_2 + \frac{1}{4}
\|\bv_{\ov{S}}\|_{2,1}.
\end{align*}
Together with \eqref{norm1} this yields
$$
\|\hat{\bx}\|_{2,1} \geq \|\bx_S\|_{2,1} - \frac{1}{4} \|\bv_S\|_2
+\frac{3}{4} \|\bv_{\ov{S}}\|_{2,1}.
$$
We now bound $\|\bv_S\|_2$. Since $\AP \bv= {\bf 0}$, we have
$(\AP)_S\bv_S= - (\AP)_{\ov{S}} \bv_{\ov{S}}$ and
\begin{align*}
\|\bv_S\|_2 &= \|[(\AP)_S^* (\AP)_S]_{|\mathcal{H}}^{-1} ({\bf
A_P})_S^* (\AP)_S \bv_S\|_2 = \|- [(\AP)_S^* ({\bf A_P})_S]_{|\mathcal{H}}^{-1}(\AP)_S^* (\AP)_{\ov{S}}
\bv_{\ov{S}}\|_2 \notag \\
&\leq \|   [(\AP)_S^* (\AP)_S]_{|\mathcal{H}}^{-1}\| \|({\bf A_P})_S^*
(\AP)_{\ov{S}} \bv_{\ov{S}}\|_2 \leq 2 \left\|
(\AP)_S^* \sum_{i \in \ov{S}} (\AP)_i v_i \right\|_2 \notag \\
&\leq 2 \sum_{i \in \ov{S}} \|(\AP)_S^*(\AP)_i\| \|v_i\|_2 \leq 2
\|\bv_{\ov{S}}\|_{2,1}.
\end{align*}
Hereby, we used the second condition in \eqref{dual_cond1}. Then we have
$$
\|\hat{\bx}\|_{2,1} \geq \|\bx\|_{2,1} + \frac{1}{4} \|\bv_{\ov{S}}\|_{2,1}.
$$
Since $\hat{\bx}$ is an
$\ell_{2,1}$-minimizer it follows that $\bv_{\ov{S}}= {\bf 0}$.
Therefore $(\AP)_S\bv_S= - (\AP)_{\ov{S}} \bv_{\ov{S}}= {\bf 0}$.
Since $(\AP)_S$ is injective, it follows that
$\bv_S={\bf 0}$, so that $\bv={\bf 0}$.
\end{proof}

To this end, we introduce the rescaled matrix $\tAP=
\frac{1}{\sqrt{m}} \AP$. The term $\|[(\tAP)_S^* (\tAP)_S]_{|\mathcal{H}}^{-1} \|$ in \eqref{dual_cond1} will be treated with
Theorem~\ref{submatrix} by noticing that $\| (\tAP)_S^* (\tAP)_S -
\bP_S \| \leq \delta$ implies $\| [(\tAP)_S^* (\tAP)_S]_{|\mathcal{H}}^{-1} \| \leq (1-\delta)^{-1}$. The other terms in
Lemma~\ref{inexact} will be estimated by the lemmas in the next
section. Throughout the proof, we use the notation $\bE_{jj}(A)$ to
denote the $s \times s$ block diagonal matrix
with the matrix $A \in \R^{d \times d}$ in its $j$-th diagonal entry
and $0$ elsewhere. 

\subsection{Auxiliary results}

We use the matrix Bernstein inequality \cite{Tropp12}, stated in
Theorem~\ref{Tropp} for convenience, in order to bound
$\|(\tAP)_S^*(\tAP)_S-\bP_S\|$. Recall the definition \eqref{definco}
of the incoherence matrix.
\begin{theorem}\label{submatrix}
Let $A \in \R^{m \times N}$ be a measurement matrix whose entries
are i.i.d. Bernoulli random variables and $(W_j)_{j=1}^N$ be a
fusion frame with the associated matrix $\Lambda$. Then, for
$\delta \in (0,1)$, the block matrix $\tAP$ satisfies
$$ \|(\tAP)_S^*(\tAP)_S-\bP_S\| \leq \delta $$
with probability at least $1-\vare$ provided
$$
m\geq \delta^{-2} \left( 2\|\Lambda^S\|_{2,\infty}^2
+ \frac{2}{3}\max\{\|\Lambda^S\|,1\} \right) \ln(2sk/\vare).
$$
\end{theorem}
\begin{proof}
We can assume that $S=[s]$ without loss of generality. Denote ${\bf
Y}_\ell = ( \eps_{\ell j} P_j )_{j\in S}$ for
$\ell \in [m]$ as the $\ell$-th block column vector of $(\tAP)_S^*$.
Observing that $\E(\bY_\ell \bY_\ell^*)_{j,k} = \E
(\eps_{\ell j}P_j \eps_{\ell k}P_k ) =
\delta_{jk}P_jP_k $, we have $\E \bY_\ell \bY_\ell^*
=\bP_S.$ Therefore, we can write
\begin{align*}
(\tAP)_S^*(\tAP)_S-\bP_S = \frac{1}{m} \sum_{\ell=1}^{m} ( \bY_\ell \bY_\ell^* - \E \bY_\ell\bY_\ell^*).
\end{align*}
This is a sum of independent self-adjoint random matrices. We denote
the block matrices $\bX_\ell := \frac{1}{m} (\bY_\ell\bY_\ell^* - \E
\bY_\ell \bY_\ell^*)$ which have mean zero. Moreover,
\begin{align*}
\|\bX_\ell \| &=\frac{1}{m} \max_{\|\bx\|_2=1,\bx \in \mathcal{H}} \left| \la
\bY_\ell \bY_\ell^* \bx,\bx \ra - \la \bP_S \bx,\bx \ra
\right| =
\frac{1}{m} \max_{\|\bx\|_2=1,\bx \in \mathcal{H}} \left| \|\bY_\ell^* \bx \|_2^2 - \| \bx \|_2^2 \right|\\
&\leq \frac{1}{m} \max \left\{ \max_{\|\bx\|_2=1,\bx \in \mathcal{H}} \|\bY_\ell^*
\bx \|_2^2 - 1 , 1 \right\} = \frac{1}{m}\max \left\{ \|
\bY_\ell \|^2 - 1, 1 \right\} .
\end{align*}
We further bound the spectral norm of
the block matrix ${\bf Y}_\ell^*$. We separate a vector $\bx \in
\R^{sd}$ into $s$ blocks of length $d$ and denote
$\bx=(x_i)_{i=1}^s$. Defining the vector $\beta \in \R^s$ with
$\beta_i = \|x_i\|_2$ we have
\begin{align}\label{basic}
\| \bY_\ell^* \|^2 &= \max_{\|\bx\|_2=1} \left\| \sum_{i
=1}^s \eps_i P_i x_i \right\|^2 = \max_{\|\bx\|_2=1}
\sum_{i,j=1}^s \eps_i \eps_j
\la P_i x_i,P_j x_j \ra  \notag \\
&\leq \max_{\|\bx\|_2=1}\sum_{i,j=1}^s | \la P_i P_j
x_j, x_i \ra | \leq \max_{\|\bx\|_2=1} \sum_{i,j=1}^s
\|P_i P_j\| \|x_i\|_2 \|x_j\|_2 \notag \\
&\leq \max_{\|\bx\|_2=1} \sum_{j=1}^s \|x_j\|_2^2 +
\max_{\|\beta\|_2=1} \sum_{i \neq j}
\|P_i P_j\| \beta_i \beta_j \notag \\
&\leq 1+ \max_{\|\beta\|_2=1} \la \beta,
\Lambda \beta \ra \leq 1+ \|\Lambda^S\|.
\end{align}
This implies the estimate
\begin{align*}
\|{\bf X}_\ell \| \leq \frac{\max\{\|\Lambda^S\|,1\}}{m}.
\end{align*}
Furthermore,
\begin{align*}
\E \bX_\ell^2 &= \frac{1}{m^2}\E \left( \bY_\ell \bY_\ell^* \bY_\ell \bY_\ell^* +
\bP_S - \bY_\ell \bY_\ell^* \bP_S -
\bP_S \bY_\ell \bY_\ell^* \right)\\
&=\E  \frac{1}{m^2} \bY_\ell \left( \sum_{j=1}^s P_j  \right)
\bY^*_\ell + \frac{1}{m^2} \bP_S - \frac{1}{m^2} \E ( \bY_\ell \bY^*_\ell )
\bP_S -
\frac{1}{m^2}\bP_S \E ( \bY_\ell \bY^*_\ell ) \\
&= \frac{1}{m^2} \sum_{i=1}^s \bE_{ii} \left( P_i \left(
\sum_{j=1}^s P_j \right) P_i \right) - \frac{1}{m^2} \bP_S.
\end{align*}
In the first equality above, we used the independence of $\eps_{\ell
j}$ for $j \in S$ and the fact that $\eps^2_{\ell j}=1$. Next, we estimate the variance parameter appearing in the noncommutative
Bernstein inequality as
\begin{align*}
\sigma^2 &:= \left\| \sum_{\ell=1}^m \E (\bX_\ell^2) \right\| =
\frac{1}{m} \left\| \sum_{i=1}^s \bE_{ii} \left( P_i \left(
\sum_{j=1}^s P_j
\right) P_i \right) - \bP_S \right\| \\
&= \frac{1}{m} \left\| \sum_{i=1}^s \bE_{ii} \left( P_i \left(\sum_{j\in[s], j \neq i} P_j
\right) P_i \right) \right\| =\frac{1}{m} \max_{i \in [s]} \left\| P_i \left(
\sum_{j \in [s], j \neq i} P_j \right) P_i \right\|.
\end{align*}
We further estimate,
\begin{align*}
  \max_{i \in [s]} \left\| P_i \left( \sum_{j \in [s], j \neq i} P_j \right) P_i
\right\| &= \max_{i \in [s]} \left\|\sum_{j \in [s], j \neq i} P_i P_j P_i
\right\| \leq \max_{i \in [s]} \sum_{\substack{j \in [s] \\  j \neq i}} \| P_i P_j \| \| P_j P_i
\| = \max_{i \in [s]} \sum_{\substack{j \in [s] \\  j \neq i}} \| P_i P_j \|^2 .
\end{align*}
Finally we arrive at
\begin{align*}
\sigma^2 \leq \frac{\|\Lambda^S\|_{2,\infty}^2 }{m}.
\end{align*}
We are now in the position of applying Theorem \ref{Tropp}. It
holds
\begin{align}\label{prob_bound}
\Prob &\left( \|(\tAP)_S^*(\tAP)_S - \bP_S \| >\delta \right) =
\Prob \left( \|
  \sum_{\ell=1}^m \bX_\ell\| >\delta \right) \notag \\
&\leq 2sk \exp \left( - \dfrac{ \delta^2 m/2}{ \|\Lambda^S\|_{2,\infty}^2
+ \max\{\|\Lambda^S\|,1\} \delta/3} \right) \leq 2sk \exp
\left( - \dfrac{ \delta^2 m}{ 2 \|\Lambda^S\|_{2,\infty}^2
+ \frac{2}{3}\max\{\|\Lambda^S\|,1\} } \right),
\end{align}
where we used that $\delta \in (0,1)$. The careful reader may have noticed that $2sk$ appears in front of the
exponential instead of the dimension of $\bX_\ell \in \R^{sd \times
sd}$ as asked by Theorem~\ref{Tropp}. In fact, Theorem~\ref{Tropp}
gives a better estimate if the matrices $\E \bX_\ell^2$
are not full rank, see \eqref{spec_tail} and the remark after
Theorem~\ref{Tropp}. Indeed, in our case since
$\rank(P_j)=\dim(W_j)=k$, we have $\rank(\E \bX_\ell^2)=sk$ which appears in \eqref{prob_bound}. Bounding the
right hand side of \eqref{prob_bound} by $\vare$ completes the
proof.
\end{proof}
We now provide the analogous of the auxiliary lemmas in \cite[Section
12.4]{Foucart13} with slight modifications.
\begin{lemma}\label{aux1}
Let $S$ be a subset of $[N]$ with cardinality $s$ and $\bv \in \R^{S
\times d}$ be a block vector of size $s$ with $v_j \in W_j$ for $j
\in S$. Assume that $m \geq \|\Lambda_S\|_{2,\infty}^2$ and $\max_{i \in
S} \|v_i\|_2 \leq \kappa \leq 1$. Then, for $t > 0$,
\begin{align*}
&\Prob \left( \max_{\ell \in \ov{S}} \| (\tAP)_\ell^* (\tAP)_S \bv
\|_2 \geq \frac{\kappa \|\Lambda_S\|_{2,\infty} }{\sqrt{m}} + t
\right)  \\
&\hspace{2in} \leq N \exp \left( - \frac{t^2 m}{2 \kappa^2
\|\Lambda_S\|_{2,\infty}^2+ 4 \kappa^2 \|\Lambda_S\|_\infty + t \kappa \|\Lambda_S\|_\infty } \right).
\end{align*}
\end{lemma}
\begin{proof}
Fix $\ell \in \ov{S}$. We may assume without loss of generality that
$S=\{1,2,\ldots, s\}$.
Observe that for $i \in [m]$, $\eps_{i \ell}$ are independent from $\eps_{ij}$ for
$j \in S$.
For simplicity we denote the corresponding  matrices as $\bB=(\tAP)_\ell^*$ and
$\bC=(\tAP)_S$. The $i$-th block column and
$i$-th block row are denoted as $\bB_i$ and $\bB^i$ respectively. Note that
\begin{align}\label{sumvec}
(\tAP)_\ell^* (\tAP)_S \bv &= \sum_{i=1}^m \bB_i \bC^i\bv =
\sum_{i=1}^m \sum_{j=1}^s \frac{1}{m} \eps_{i \ell} \eps_{i j}
P_\ell P_j v_j
\end{align}
for $\ell \in \ov{S}$. For convenience we introduce
\begin{align*}
\bB_i \bC^i= \frac{1}{m}
\begin{pmatrix}
\eps_{i \ell} \eps_{i 1} P_\ell P_1 & \eps_{i \ell} \eps_{i 2}
P_\ell P_2 & \cdots & \eps_{i \ell} \eps_{i s} P_\ell P_s
\end{pmatrix}.
\end{align*}
The sum of independent vectors in \eqref{sumvec} will be bounded in
$\ell_2$ norm using the vector valued Bernstein inequality
Lemma~\ref{v_bern}. Observe that the vectors $\bB_i \bC^i \bv$ have
mean zero. Furthermore,
\begin{align*}
&m \E \|\bB_i \bC^i \bv\|_2^2 = \frac{1}{m} \E \sum_{j,k=1}^s
\eps_{i \ell}^2 \eps_{ij} \eps_{ik} \la P_\ell P_j v_j, P_\ell P_k
v_k \ra \\
&= \frac{1}{m} \sum_{j=1}^s \|P_\ell P_jv_j\|_2^2 \leq \frac{1}{m}
\sum_{j=1}^s \|P_\ell P_j\|^2 \|v_j\|_2^2 \leq \frac{\kappa^2}{m}
\|\Lambda_S\|_{2,\infty}^2
\end{align*}
where we used $\|v_j\|_2 \leq \kappa$. We bound $\sigma^2$ appearing
in Lemma~\ref{v_bern} simply due to Remark~\ref{variance} by
$$
m \sigma^2 \leq m \E \|\bB_i \bC^i \bv\|_2^2 \leq \frac{\kappa^2}{m}
\|\Lambda_S\|_{2,\infty}^2.
$$
For the uniform bound, observe that
\begin{align*}
\|\bB_i \bC^i \bv\|_2 &= \frac{1}{m} \left\| \sum_{j=1}^s
\eps_{i \ell} \eps_{i j} P_\ell P_j v_j \right\|_2 \leq \frac{1}{m}
\sum_{j=1}^s \|P_\ell P_j\| \|v_j\|_2 \leq \frac{\kappa}{m}
\|\Lambda_S\|_\infty.
\end{align*}
Then the vector valued Bernstein inequality \eqref{bern3} yields
\begin{align*}
&\Prob \left( \| (\tAP)_\ell^* (\tAP)_S \bv \|_2 \geq \frac{\kappa
\|\Lambda_S\|_{2,\infty} }{\sqrt{m}} + t \right) \\
& \hspace{2in} \leq \exp \left(
- \frac{t^2/2}{\frac{\kappa^2\|\Lambda_S\|_{2,\infty}^2 }{m}  + \frac{2
\kappa \|\Lambda_S\|_\infty}{m} \frac{\kappa
\|\Lambda_S\|_{2,\infty} }{\sqrt{m}} + \frac{t}{3} \frac{\kappa
\|\Lambda_S\|_\infty}{m}} \right).
\end{align*}
Taking the union bound over $\ell \in \ov{S} \subset [N]$ and using
that $\frac{\|\Lambda_S\|_{2,\infty} }{\sqrt{m}} \leq 1$  yields
\begin{align*}
&\Prob \left( \max_{\ell \in \ov{S}} \| (\tAP)_\ell^* (\tAP)_S \bv
\|_2 \geq \frac{\kappa \|\Lambda_S\|_{2,\infty} }{\sqrt{m}} + t
\right) \\
&\hspace{2in} \leq N \exp \left( - \frac{t^2 m}{2 \kappa^2
\|\Lambda_S\|_{2,\infty}^2+ 4 \kappa^2 \|\Lambda_S\|_\infty + t \kappa \|\Lambda_S\|_\infty } \right).
\end{align*}
This completes the proof.
\end{proof}

Next, we prove a similar auxiliary result.

\begin{lemma}\label{aux2} Let $S$ be subset of $[N]$ with cardinality $s$ and $\bv \in \R^{S
\times d}$ be a block vector of size $s$ with $v_j \in W_j$ for $j
\in S$. Assume that $m \geq \|\Lambda^S\|_{2,\infty}^2$. Then, for $t>0$,
\begin{align*}
&\Prob \left( \| [(\tAP)_S^*(\tAP)_S- \bP_S] \bv\|_2 \geq \left(
\frac{ \|\Lambda^S\|_{2,\infty} }{\sqrt{m}}+ t \right) \|\bv\|_2
\right) \\
&\hspace{2in} \leq \exp \left(-\frac{mt^2}{8+ 4 \|\Lambda^S\|_\infty + 2 \|\Lambda^S\|_{2,\infty}^2+ t
(\frac{4}{3} +\frac{2}{3} \|\Lambda^S\|_\infty)} \right).
\end{align*}
\end{lemma}
\begin{proof}
Again we assume without loss of generality that $\|\bv\|_2=1$ and
$S=\{1,2,\ldots, s\}$. As in the proof of Theorem~\ref{submatrix},
we rewrite the term that we need to bound as
\begin{align}\label{ind_sum}
[(\tAP)_S^*(\tAP)_S- \bP_S] \bv = \frac{1}{m} \sum_{\ell=1}^m
(\bY_\ell \bY_\ell^*-\bP_S) \bv
\end{align}
where $\bY_\ell = (\eps_{\ell i}P_i)_{i=1}^s$ is the $\ell$-th block row
of $(\tAP)_S$. We use the vector valued Bernstein
inequality, Lemma~\ref{v_bern} once again, in order to estimate the
$\ell_2$ norm of this sum. Observe that $\E (\bY_\ell
\bY_\ell^*-\bP_S)\bv ={\bf 0}$ as in the proof of Theorem~\ref{submatrix}.
Furthermore, denoting
\begin{align*}
Z= \left\| \frac{1}{m} \sum_{\ell=1}^m (\bY_\ell \bY_\ell^*-\bP_S)
\bv \right\|_2,
\end{align*}
we have
\begin{align*}
\E Z^2 &= m \E \left\| \frac{1}{m} (\bY_1 \bY_1^*-\bP_S)\bv
\right\|_2^2 \\
&= \frac{1}{m} \E \la (\bY_1 \bY_1^*-\bP_S)\bv, (\bY_1
\bY_1^*-\bP_S)\bv \ra \\
&=\frac{1}{m} \E (\la \bY_1 \bY_1^*\bv , \bY_1
\bY_1^*\bv \ra -2 \la
\bY_1 \bY_1^*\bv,\bv \ra + \la \bv,\bv\ra ) \\
&=\frac{1}{m} \E (\la \bY_1 \bY_1^* \bY_1 \bY_1^*\bv ,
\bv \ra -2 \| \bY_1^*\bv\|_2^2 + 1 ).
\end{align*}
We now estimate the first two terms in the last line above.
First observe that due to $\bY_1^*\bY_1= \sum_{i=1}^s P_i$, it holds
\begin{align*}
\E \la \bY_1 \bY_1^* \bY_1 \bY_1^*\bv , \bv \ra &= \la
\E(\bY_1 \sum_{i=1}^s P_i \bY_1^* )\bv,\bv \ra = \la \sum_{j=1}^s \bE_{jj} ( P_j \sum_{i=1}^s P_i P_j )
\bv,\bv \ra \\
&= \sum_{j=1}^s \la P_j \sum_{i=1}^s P_i P_j v_j,v_j \ra
=\sum_{i,j=1}^s \la P_jP_i P_j v_j,v_j \ra  \\
&\leq \sum_{i,j, i \neq j} \|P_i P_j\|^2 \|v_j\|_2^2 + \sum_{j=1}^s \|v_j\|_2^2 \leq \left( \sum_{j=1}^s \|v_j\|_2^2 \sum_{i \neq j} \|P_i P_j\|^2
\right) +1 \\
&\leq 1 + \left(\max_{j \in S} \sum_{i \neq j} \|P_i P_j\|^2 \right)\sum_{j=1}^s \|v_j\|_2^2 \leq 1+ \|\Lambda^S\|_{2,\infty}^2,
\end{align*}
where we used that $\Lambda^S$ is symmetric and $\|\bv\|_2=1$.
Secondly, since
\begin{align*}
\E\|\bY_1^*\bv\|_2^2 &= \E \| \sum_{i=1}^s \eps_i P_i v_i\|_2^2
=\E \sum_{i,j=1}^s \eps_i \eps_j \la v_i,v_j \ra = \sum_{i=1}^s
\|v_i\|_2^2 =1,
\end{align*}
 we obtain
\begin{align*}
\E Z^2 \leq \frac{\|\Lambda^S\|_{2,\infty}^2 }{m}.
\end{align*}
For the uniform bound, we have
\begin{align*}
\frac{1}{m}\|(\bY_\ell \bY_\ell^*-\bP_S)\bv\|_2 &\leq \frac{1}{m}
\|\bY_\ell
\bY_\ell^*\| \|\bv\|_2 + \frac{1}{m} \|\bv\|_2 \\
&= \frac{1}{m} \|\bY_\ell\|^2 + \frac{1}{m} \leq \frac{2 +
\|\Lambda^S\|_\infty}{m}.
\end{align*}
The last inequality follows from \eqref{basic}.
Finally we estimate the weak variance simply by the strong variance
\begin{align*}
m \sigma^2 \leq \E Z^2 \leq \frac{\|\Lambda^S\|_{2,\infty}^2}{m}.
\end{align*}
Then the $\ell_2$-valued Bernstein inequality \eqref{bern3} yields
\begin{align*}
&\Prob \left( \| [(\tAP)_S^*(\tAP)_S- \bP_S] \bv\|_2 \geq \left(
\frac{ \|\Lambda^S\|_{2,\infty} }{\sqrt{m}} + t \right) \|\bv\|_2
\right) \\
&\hspace{2in}  \leq \exp
\left(-\frac{t^2/2}{\frac{ \|\Lambda^S\|_{2,\infty}^2 }{m}+\frac{(4+2\|\Lambda^S\|_\infty)}{m}\frac{\|\Lambda^S\|_{2,\infty}}{\sqrt{m}}
+ \frac{t (2+\|\Lambda^S\|_\infty)}{3m}} \right).
\end{align*}
Using that $\frac{\|\Lambda^S\|_{2,\infty}}{\sqrt{m}} \leq 1$, we
obtain
\begin{align*}
&\Prob \left( \| [(\tAP)_S^*(\tAP)_S- \bP_S] \bv\|_2 \geq \left(
\frac{ \|\Lambda^S\|_{2,\infty} }{\sqrt{m}}+ t \right) \|\bv\|_2
\right) \\
& \hspace{2in} \leq \exp \left(-\frac{mt^2}{8+ 4 \|\Lambda^S\|_\infty + 2 \|\Lambda^S\|_{2,\infty}^2+ t
(\frac{4}{3} +\frac{2}{3} \|\Lambda^S\|_\infty)} \right).
\end{align*}
This completes the proof.
\end{proof}

Lemma~\ref{aux2} shows that the multiplication with $(\tAP)_S^*(\tAP)_S- \bP_S$ decreases
the $\ell_2$ norm of the vectors with high probability. The next
lemma shows that this is true for $\ell_{2,\infty}$ norm as well.

\begin{lemma}\label{aux22} Assume the conditions of Lemma~\ref{aux2}. Then, for $t>0$,
\begin{align*}
&\Prob \left( \|[(\tAP)_S^*(\tAP)_S- \bP_S] \bv\|_{2,\infty} \geq
\left( \frac{ \|\Lambda^S\|_{2,\infty} }{\sqrt{m}}+ t \right)
\|\bv\|_{2,\infty} \right) \\
&\hspace{2in} \leq
s \cdot  \exp \left(-\frac{mt^2}{4 \|\Lambda^S\|_\infty + 2
    \|\Lambda^S\|_{2,\infty}^2 + \frac{2}{3} t \|\Lambda^S\|_\infty}\right).
\end{align*}
\end{lemma}

\begin{proof}
We can assume that $S= [s]$ and $\|\bv\|_{2,\infty} = \max_{i \in S}
\|v_i\|_2 =1$ by normalizing $\bv$ by  $\|\bv\|_{2,\infty}$. As in
\eqref{ind_sum}, we write
\begin{align*}
\bZ:=[(\tAP)_S^*(\tAP)_S- \bP_S] \bv = \frac{1}{m} \sum_{\ell=1}^m
(\bY_\ell \bY_\ell^*-\bP_S) \bv
\end{align*}
where $\bY_\ell = (\eps_{\ell i}P_i)_{i=1}^s$ is the $\ell$-th row
of $(\tAP)_S$. We can further write for $i \in S$
$$ Z_i = \sum_{\ell=1}^s \frac{1}{m} \sum_{\substack{j=1 \\ j \neq
    i}}^s \eps_{\ell i} \eps_{\ell j} P_i P_j v_j =: \sum_\ell
X_\ell. $$ The vectors $X_\ell$ are independent, thus we use
Lemma~\ref{v_bern} in order to bound $\|Z_i\|_2$. Since we have done
similar estimations in the previous proofs, we skip some steps and obtain
\begin{align*}
\E  \|Z_i\|_2^2 &= m \E \| X_\ell\|_2^2 \leq \frac{1}{m}
\sum_{\substack{j=1 \\ j \neq i}}^s \|P_i P_j\|^2 \|v_j\|_2^2 \leq \frac{1}{m} \|\Lambda^S\|_{2,\infty}^2,
\end{align*}
where we used that $\|v_j\|_2 \leq 1$. Furthermore, $m \sigma^2 \leq
\frac{1}{m} \|\Lambda^S\|_{2,\infty}^2$. For any $\ell \in [s]$ we
have the uniform bound
\begin{align*}
\|X_\ell\|_2 &= \frac{1}{m} \left\| \sum_{j=1,j \neq i}^s \eps_{\ell i} \eps_{\ell j} P_i P_j v_j \right\|_2 \leq \frac{1}{m} \sum_{j=1,j \neq i}^s \|P_i P_j \|
\|v_j\|_2 \leq \frac{1}{m} \|\Lambda^S\|_\infty.
\end{align*}
Combining these with Lemma \ref{v_bern} and taking the union bound yield
$$
\Prob \left( \max_{i \in S} \|Z_i\|_2 \geq \frac{
\|\Lambda^S\|_{2,\infty} }{\sqrt{m}}+ t \right) \leq s \cdot
\exp \left(-\frac{mt^2}{4 \|\Lambda^S\|_\infty + 2
    \|\Lambda^S\|_{2,\infty}^2 + \frac{2}{3} t
\|\Lambda^S\|_\infty}\right).
$$
\end{proof}

Lastly we present the following lemma before the proof of our main result.

\begin{lemma}\label{aux3} For $t \in (0,\frac{3}{2})$,
$$\Prob ( \max_{i \in \ov{S}} \| (\tAP)_S^* (\tAP)_i\| \geq t) \leq
2(s+1)Nk \exp\left(-\frac{t^2 m}{3 \|\Lambda_S\|_{2,\infty}^2 } \right).$$
\end{lemma}
\begin{proof}
Fix $i \in \ov{S}$. Similarly as before, we write $(\tAP)_S^*
(\tAP)_i$ as a sum of independent matrices,
\begin{align}\label{recsum}
(\tAP)_S^* (\tAP)_i
=\frac{1}{m} \sum_{\ell=1}^m \begin{pmatrix}
\eps_{\ell 1} \eps_{\ell i} P_1 P_i \\ \eps_{\ell 2}\eps_{\ell i} P_2 P_i  \\  \vdots \\
\eps_{\ell s}\eps_{\ell i} P_s P_i
\end{pmatrix} =: \frac{1}{m} \sum_{\ell=1}^m \bY_\ell,
\end{align}
where we assumed $S = [s]$ for simplifying the notation. 
Above we introduced the block column vectors $\bY_\ell \in \R^{sd
\times d}$ which are independent and identically distributed
rectangular matrices. Observe also that $\E \bY_\ell ={\bf 0}$. In order
to estimate the norm of the sum in \eqref{recsum} we will employ
Theorem \ref{Tropp_rect} \cite{Tropp12}  which is a
version of the noncommutative Bernstein inequality for rectangular
matrices. We first bound the
variance parameter
\begin{align}\label{last_var}
\sigma^2=  \max \left\{ \|\sum_{\ell=1}^m \frac{1}{m^2} \E[\bY_\ell
\bY_\ell^*] \|, \|\sum_{\ell=1}^m \frac{1}{m^2} \E[\bY_\ell^*
\bY_\ell] \| \right\}.
\end{align}
We write $\E[\bY_\ell \bY_\ell^*] = \sum_{j=1}^s \bE_{jj} ( P_j P_i P_j ).$
The first term on the right hand side of \eqref{last_var} is estimated as
\begin{align}\label{former_term}
\|\frac{1}{m^2}\sum_{\ell=1}^m  \E[\bY_\ell \bY_\ell^*] \| =
\frac{1}{m} \max_{j \in S} \|P_j P_i P_j\| \leq \frac{1}{m} \max_{j \in S} \|P_j P_i\| \|P_i P_j\| \leq
\frac{\lambda^2}{m}.
\end{align}
We used that $P_i^2= P_i$ and $i \not\in S$. Furthermore, $\E[\bY_\ell^* \bY_\ell] = \sum_{j=1}^s P_i P_j P_i$ and
\begin{align}\label{latter_term}
\frac{1}{m^2} \|\sum_{\ell=1}^m  \E[\bY_\ell^* \bY_\ell] \| &=
\frac{1}{m} \| \sum_{j=1}^s P_i P_j P_i \| \leq \frac{1}{m} \sum_{j=1}^s \| P_i P_j\| \|P_j P_i \| \leq
\frac{\|\Lambda_S\|_{2,\infty}^2 }{m}.
\end{align}
Since \eqref{latter_term} dominates \eqref{former_term}, we have
$$\sigma^2 \leq \frac{\|\Lambda_S\|_{2,\infty}^2 }{m}.$$
For the uniform bound we obtain
\begin{align*}
\|\bY_\ell\|^2 &= \sup_{\substack{ \|x\|_2\leq 1 \\
x\in \R^d }} \| \bY_\ell x \|_2^2 = \sup_{\substack{ \|x\|_2\leq 1 \\
x\in \R^d }} \sum_{j=1}^s \|P_jP_i
x \|_2^2 \\
&\leq \sup_{\substack{ \|x\|_2\leq 1 \\
x\in \R^d }} \sum_{j=1}^s \|P_jP_i\|^2 \|x \|_2^2 \leq
\|\Lambda_S\|_{2,\infty}^2.
\end{align*}
We conclude that $\frac{1}{m}\|\bY_\ell\| \leq \frac{
\|\Lambda_S\|_{2,\infty}}{m}$. Combining these estimates, Theorem \ref{Tropp_rect} yields
$$
\Prob (\| (\tAP)_S^* (\tAP)_i\| \geq t) \leq 2(s+1) k
\exp\left(-\frac{t^2/2}{\frac{\|\Lambda_S\|_{2,\infty}^2}{m} +
\frac{t}{3}\frac{\|\Lambda_S\|_{2,\infty} }{m}} \right).
$$
Taking the union bound over $i \in \ov{S} \subset [N]$ and using
that $t \in (0,\frac{3}{2})$ yields
$$\Prob ( \max_{i \in \ov{S}} \| (\tAP)_S^* (\tAP)_i\| \geq t) \leq
2(s+1)Nk \exp\left(-\frac{t^2 m}{3 \|\Lambda_S\|_{2,\infty}^2 } \right).$$
Above, the dimension of the subspaces $k$ appears instead the ambient
dimension $d$ for the same reasons explained in the proof of Theorem~\ref{submatrix}. This completes the
proof.
\end{proof}

\subsection{Proof of Theorem~\ref{nonuni_main}}\label{golf_proof}
Essentially we follow the arguments in \cite[Section 12.4]{Foucart13}.
We will construct an inexact dual vector as in Lemma~\ref{inexact} satisfying the
conditions there. To this end, we will use the so-called
\textit{golfing scheme} due to Gross \cite{Gross09}.  We partition
the $m$ independent (block) rows of $\AP$ into $L$ disjoint blocks
of sizes $m_1, \ldots, m_L$ and $L$ to be specified later with
$m=\sum_{j=1}^L m_j$. These blocks correspond to row submatrices of
$\AP$ which are denoted by $\AP^{(1)} \in \R^{m_1 d \times Nd},
\ldots,\AP^{(L)} \in \R^{m_L d \times Nd}$, i.e.,
\begin{align*}
{\bf A_P}=
\begin{pmatrix}
& & & \AP^{(1)}& & & \\
& & & \AP^{(2)}& & & \\
& & & \vdots & & & \\
& & &  \AP^{(L)}& & &\end{pmatrix}
\begin{array}{l}
\} m_1 \\ \} m_2 \\ \vdots \\ \} m_L \end{array}
\end{align*}
Set $S=\supp (x)$. The golfing scheme starts with $\bu^{(0)}={\bf 0}$ and
then inductively defines
\begin{align*}
\bu^{(n)}= \frac{1}{m_n} (\AP^{(n)})^* (\AP^{(n)})_S (\sgn(\bx_S) -
\bu_S^{(n-1)}) + \bu^{(n-1)},
\end{align*}
for $n=1,\ldots,L$. The vector $\bu=\bu^{(L)}$ will serve as a
candidate for the inxeact dual vector in Lemma~\ref{inexact}. Thus,
we need to check if it satisfies the two conditions in
\eqref{dual_cond2}. By construction $\bu$ is in the row space of
$\AP$, i.e., $\bu=\AP^* \bh$ for some vector $\bh$ as required in
Lemma~\ref{inexact}. To simplify the notation we introduce
$\bw^{(n)}=\sgn(\bx_S)-\bu_S^{(n)}$. Observe that
\begin{align}\label{star1}
\bu_S^{(n)} - \bu_S^{(n-1)} &= \frac{1}{m_n} (\AP^{(n)})_S ^*
(\AP^{(n)})_S  (\sgn(\bx_S) - \bu_S^{(n-1)}) \notag \\
\bw^{(n-1)}-\bw^{(n)} &= \frac{1}{m_n} (\AP^{(n)})_S ^*
(\AP^{(n)})_S
\bw^{(n-1)} \notag \\
\bw^{(n)} &= \left[ \bP_S- \frac{1}{m_n} (\AP^{(n)})_S ^*
(\AP^{(n)})_S \right] \bw^{(n-1)}
\end{align}
Above we used that $\bP_S \bw^{(n)} = \bw^{(n)}$. Furthermore we have
\begin{align}\label{star2}
\bu^{(n)} - \bu^{(n-1)} &= \frac{1}{m_n} (\AP^{(n)})^*
(\AP^{(n)})_S \bw^{(n-1)} \notag \\
\bu=\bu^{(L)} &= \sum_{n=1}^L \frac{1}{m_n} (\AP^{(n)})^*
(\AP^{(n)})_S \bw^{(n-1)},
\end{align}
where last line follows by a telescopic sum.
We will later show that the matrices
$$
\bP_S- \frac{1}{m_k}(\AP^{(k)})_S^* (\AP^{(k)})_S
$$
are contractions and the norm of the
residual vector $\bw^{(n)}$ decreases geometrically fast, thus
$\bu^{(n)}$ becomes close to $\sgn(\bx_S)$ on its support set $S$.
Particularly, we will prove that $\|\bw^{(L)}\|_2 \leq 1/4$ for a
suitable choice of $L$. In
addition we also need that the off-support part of $\bu$ remains
small as well, satisfying the condition $\max_{i \in \ov{S}} \|u_i\|_2
\leq 1/4$. 

For these tasks, we will use the lemmas proven above. For the
moment, we assume that the following holds for each $n$ with high
probability
\begin{align}\label{est0}
\|\bw^{(n)}\|_{2,\infty} \leq \left(
\frac{\|\Lambda^S\|_{2,\infty} }{\sqrt{m}} + q_n \right)
\|\bw^{(n-1)}\|_{2,\infty}, \ \ n \in [L].
\end{align}
Let $q_n' := \frac{ \|\Lambda^S\|_{2,\infty} }{\sqrt{m}} + q_n$.
Since $\|\bw^{(0)}\|_{2,\infty} = \|\sgn(\bx_S)\|_{2,\infty} = 1$,
we have
\begin{align*}
\|\bw^{(n)}\|_{2,\infty} \leq \prod_{j=1}^n q_j' =: h_n.
\end{align*}
Further assume that the following inequalities hold for each $n$
with high probability,
\begin{align}
& \|\bw^{(n)}\|_2 \leq \left(\frac{ \|\Lambda^S\|_{2,\infty} }{\sqrt{m}} + r_n \right) \|\bw^{(n-1)}\|_2, \ \ n \in [L], \label{est1} \\
&\max_{i \in \ov{S}} \left\| \frac{1}{m_n} (\AP^{(n)})_i^* (\AP^{(n)})_S \bw^{(n-1)} \right\|_2 \leq \frac{h_n \|\Lambda_S\|_{2,\infty} }{\sqrt{m}} + t_n , \ \ n \in [L]. \label{est2}
\end{align}
The parameters $q_n, r_n, t_n$ will be specified later. Now let
$r_n' := \frac{ \|\Lambda^S\|_{2,\infty} }{\sqrt{m}} + r_n$ and
$t_n' := \frac{h_n \|\Lambda_S\|_{2,\infty} }{\sqrt{m}}+t_n$. Then
the relations in \eqref{star1} and \eqref{est1} yield
\begin{align*}
\|\sgn(\bx_S) - \bu_S \|_2 = \|\bw^{(L)}\|_2 \leq \|\sgn(\bx_S)\|_2
\prod_{n=1}^L r_n' \leq \sqrt{s} \prod_{n=1}^L r_n'.
\end{align*}
Furthermore, \eqref{star2} and \eqref{est2} give
\begin{align*}
\max_{i \in \ov{S}} \|u_i\|_2 &= \max_{i \in \ov{S}} \left\|
\sum_{n=1}^L \frac{1}{m_n} (\AP^{(n)})_i^* (\AP^{(n)})_S \bw^{(n-1)}
\right\|_2 \notag \\
&\leq \sum_{n=1}^L \max_{i \in \ov{S}} \left\| \frac{1}{m_n}
(\AP^{(n)})_i^* (\AP^{(n)})_S \bw^{(n-1)} \right\|_2 \notag \\
& \leq \sum_{n=1}^L t_n'.
\end{align*}
Next we define the probabilities $p_0(n)$, $p_1(n)$ and $p_2(n)$ that \eqref{est0}, \eqref{est1} and
\eqref{est2} do not hold respectively. Then by Lemma~\ref{aux22} and independence of the
blocks,
$$p_0(n) \leq \vare,$$
provided
\begin{align}\label{m_cond0}
m_n \geq \left( \frac{4 \|\Lambda^S\|_\infty + 2 \|\Lambda^S\|_{2,\infty}^2 }{q_n^2}+\frac{2 \|\Lambda^S\|_\infty}{3 q_n} \right) \ln(s/\vare).
\end{align}
Also by Lemma~\ref{aux2} and independence of the blocks,
$$p_1(n) \leq \vare$$
provided
\begin{align}\label{m_cond1}
m_n \geq \left( \frac{8+ 4 \|\Lambda^S\|_\infty + 2 \|\Lambda^S\|_{2,\infty}^2 }{r_n^2}+\frac{4/3+ (2/3)
\|\Lambda^S\|_\infty}{r_n} \right) \ln(\vare^{-1}).
\end{align}
Similarly, due to Lemma~\ref{aux1} and independence of the blocks,
$$p_2(n) \leq \vare,$$
provided
\begin{align}\label{m_cond2}
m_n \geq \left( \frac{ 2 h_n^2 \|\Lambda_S\|_{2,\infty}^2 + 4 h_n^2 \|\Lambda_S\|_\infty}{t_n^2}+\frac{h_n
\|\Lambda_S\|_\infty}{t_n} \right) \ln(N/\vare).
\end{align}
We now set the parameters $L, m_n, t_n, r_n, q_n$ for $n \in
[L]$ such that $\|\sgn(\bx_S)-\bu_S\|_2 \leq 1/4$ and $\max_{i \in
\ov{S}} \|u_i\|_2 \leq 1/4$ as required in the Lemma~\ref{inexact}.
We choose
\begin{align*}
&L=\lceil \ln(s)/\ln\ln (N) \rceil +3, \\
& m_n \geq c (1+\|\Lambda_S\|_\infty) \ln(N) \ln(2L \vare^{-1}), \\
&r_n=\frac{1}{4\sqrt{\ln(N)}},\\
& t_n=\frac{1}{2^{n+3}}, \\
& q_n =\frac{1}{8}.
\end{align*}
We can estimate each of $\|\Lambda^S\|_{2,\infty}^2, \|\Lambda_S\|_{2,\infty}^2,  \|\Lambda^S\|_\infty$ by $ \|\Lambda_S\|_\infty$ from above. Then by definitions of $r_n',t_n',h_n,q_n'$, we obtain $r_n'\leq
\frac{1}{2\sqrt{\ln N}}$, $q_n' \leq \frac{1}{2}$, $h_n \leq
\frac{1}{2^n}$ and $t_n' \leq \frac{1}{2^{n+2}}$ for $n=1,\ldots,L$.
Furthermore,
\begin{align*}
\|\sgn(\bx_S)-\bu_S\|_2 \leq \sqrt{s} \prod_{n=1}^L r_n' \leq
\frac{1}{4},
\end{align*}
and
\begin{align*}
\max_{i \in \ov{S}} \|u_i\|_2 \leq \sum_{n=1}^L t_n' \leq
\frac{1}{4}.
\end{align*}
Next we bound the failure probabilities according to our
choices of parameters above. Considering also the conditions
\eqref{m_cond0}, \eqref{m_cond1} and \eqref{m_cond2}, we have
$p_0(n),p_1(n),p_2(n) \leq \vare/ L.$ These yield
$$\sum_{n=1}^L p_0(n) + p_1(n) +p_2(n) \leq 3 \vare.$$
The overall number of samples obey
\begin{align}\label{imp1}
m = \sum_{n=1}^L m_n \geq 2c(1+\|\Lambda_S\|_\infty) L
\ln(N)\ln(L/\vare).
\end{align}
This
is already very close to the proposed condition in the statement of
our theorem. We will strengthen this condition later. Next we look
into the first part of Condition \eqref{dual_cond1} of
Lemma~\ref{inexact}. By Theorem~\ref{submatrix},
$\|(\tAP)_S^*(\tAP)_S-\bP_S\| \leq 1/2$ with probability at least
$1-\vare$ provided
\begin{align}\label{imp2}
m\geq \left( 8\|\Lambda^S\|_{2,\infty}^2
+ \frac{8}{3}\max\{\|\Lambda^S\|,1\} \right) \ln(2sk/\vare)..
\end{align}
This implies that $\|[(\tAP)_S^* (\tAP)_S]_{|\mathcal{H}}^{-1} \| \leq
2$. For the second part of Condition \eqref{dual_cond1} we will use
Lemma~\ref{aux3}. It says that
$$
\Prob ( \max_{i \in \ov{S}} \| (\tAP)_S^* (\tAP)_i\| \geq t) \leq
2(s+1)Nk \exp\left(-\frac{t^2 m}{3 \|\Lambda_S\|_{2,\infty}^2 } \right).
$$
Taking $t=1$ implies
that
$$
\max_{i \in \ov{S}} \| (\tAP)_S^* (\tAP)_i\| \leq 1
$$
with probability at least $1-\vare$ provided
\begin{align}\label{imp3}
m \geq 6 \|\Lambda_S\|_{2,\infty}^2  \ln(N(s+1)k/\vare).
\end{align}
Altogether we have shown that Conditions \eqref{dual_cond1} and
\eqref{dual_cond2} of Lemma~\ref{inexact} hold simultaneously with
probability at least $1-5\vare$ provided Conditions \eqref{imp1},
\eqref{imp2} and \eqref{imp3} hold. Replacing $\vare$ by $\vare/5$,
the main condition of Theorem~\ref{nonuni_main}
$$
m \geq C (1 + \|\Lambda_S\|_\infty ) \ln( N) \ln(sk) \ln(\vare^{-1})
$$
implies all three
conditions above with an appropriate constant $C$ since $\|\Lambda^S\| \leq \|\Lambda^S\|_\infty \leq \|\Lambda_S\|_\infty$ since $\Lambda$ is symmetric.

This ends the proof of our theorem.

\qed


\begin{remark}
The inexact dual method yields a relatively long and technical proof for sparse recovery and involves several auxiliary results.
Other methods used in the compressed sensing literature proved to be hard to apply for our particular case where we work with the block matrix $\AP$ 
which is more structured than a purely random Gaussian matrix.
To name a few of other methods, the exact dual approach developed by J. J. Fuchs  
\cite{Fuchs04} was used for subgaussian matrices in \cite{Ayaz11}, a uniform recovery result for Gaussian matrices
based on the concentration of measure of Lipschitz functions was given by one of the seminal papers by Cand\'{e}s and 
Tao \cite{cata06} and atomic norm approach that was recently given in \cite{Chandra12} 
with far-reaching applications. For instance, particularly the exact dual approach 
involves taking the pseudo-inverse of $\AP$ which loses the structure given by projection matrices $P_i$. 
This structure is crucial because it allows us to prove our results involving the incoherence parameter $\lambda$ 
which is the central theme of this paper. 
\end{remark}

\subsection{A special case}

In this section we give an alternative result for the nonuniform
recovery with Bernoulli matrices. This result involves the parameter
$\lambda$ instead of matrix $\Lambda$.

\begin{theorem}\label{nonuni_corol}Let $\bx \in \mathcal{H}$ be $s$-sparse. Let $A \in \R^{m \times N}$ be Bernoulli matrix
and $(W_j)_{j=1}^N$ be given with parameter $\lambda \in [0,1]$.
Assume that
\begin{align}\label{nonuni_corol_cond}
m \geq C (1 + \lambda s) \ln(Nsk) \ln(\vare^{-1}),
\end{align}
where $C > 0$ is a universal constant. Then with probability at
least $1-\vare$, $(L1)$ recovers $\bx$ from $\by=\AP \bx$.
\end{theorem}

The proof of this theorem is similar to the one of
Theorem~\ref{nonuni_main} with slight modifications in the
estimations.

\begin{remark}
Theorem~\ref{nonuni_corol} improves Theorem~\ref{nonuni_main} in
terms of the log-factors as $\log(s)$ does not appear in Condition
\eqref{nonuni_corol_cond}. 
Condition~\eqref{nonuni_cond} is slightly better than
\eqref{nonuni_cond} in terms of the incoherence parameter, at
least if there is a true gap between $\|\Lambda\|_\infty$ and $\lambda
s$, which happens if the quantities $\|P_i P_j\|$ are not all close to
their maximal value. The equality is achieved when the subspaces are equi-angular. In the case that they are not equi-angular,
even if only two subspaces align, then $\lambda=1$. In
this case, \eqref{nonuni_corol_cond} suggests that we should not
expect any improvement for the recovery of the sparse vectors with
respect to the standard block sparse case. However, intuitively the
orientation of the other subspaces might still be effective in the
recovery process. A more average measure of incoherence of the
subspaces is captured by $\|\Lambda\|_\infty$ in
\eqref{nonuni_cond}, so Theorem~\ref{nonuni_main} improves for
general orientations of the subspaces up to a slight drawback in the
log-factors. Numerical experiments we have run also support this
result.
\end{remark}


\section{Stable and Robust Recovery}\label{stablesec}

In this section we show that nonuniform recovery for fusion frames
with Bernoulli matrices are stable and robust under
presence of noise. In other words we allow our signal $\bx$ to be
approximately sparse (compressible) and the measurements
$\by$ to be noisy. Our measurement model then becomes
\begin{align}\label{noisy_model}
\by = \AP \bx + \be \ \text{ with } \ \|\be\|_2 \leq \eta \sqrt{m}
\end{align}
for some $\eta \geq 0$. For the reconstruction we employ
\begin{align*}
(L1)^\eta \ \ \hat{\bx}= \text{argmin}_{x \in \mathcal{H}}
\|\bx\|_{2,1} \ \ s.t. \ \ \|\AP \bx - \by\|_2 \leq \eta \sqrt{m}.
\end{align*}

The condition $\|\be\|_2 \leq \eta \sqrt{m}$ in \eqref{noisy_model} is
natural for a vector $\be = (e_j)_{j=1}^m$. For instance, it is implied
by the bound $\|e_j\|_2 \leq \eta$ for all $j \in [m]$. 
We first define the best $s$-term approximation of a vector $\bx$ as follows
$$
\sigma_s(\bx)_1 := \inf_{\|\bz\|_0 \leq s} \|\bx -\bz \|_{2,1}.
$$
Compressible vectors are the ones with small $\sigma_s(\bx)_1$. The
next statement makes Lemma~\ref{inexact} stable and robust under
noise and under passing from sparse to compressible vectors. It is
an extension of \cite[Theorem 4.33]{Foucart13} and its proof is entirely
analogous to the one in there, so we skip it.

\begin{lemma}\label{inexact_robust}
Let $A \in \R^{m \times N}$ and $(W_j)_{j=1}^N$ be a fusion frame
for $\R^d$ and $x \in \mathcal{H}$. Let $S \subset [N]$ be the index
set of the $s$ largest $\ell_2$-normed vectors $x_i$ of $\bx$.
Assume that, for positive constants $\delta, \beta, \gamma, \theta
\in (0,1)$ with $b := \theta + \beta \gamma / (1-\delta) < 1$ and
\begin{align}
\|(\AP)_S^* \AP_S - \bP_S\| &\leq \delta, \label{robust_cond1} \\
\max_{\ell \in \ov{S}} \|(\AP)_S^* (\AP)_\ell \| &\leq \beta.
\label{robust_cond2}
\end{align}
Suppose there exists a block vector $\bu \in \R^{Nd}$ of the form
$\bu = \AP^* \bh$ with block vector $\bh \in \R^{md}$ such that
\begin{align}
\|\bu_S- \sgn(\bx_S)\|_2 &\leq \gamma, \label{robust_cond3} \\ \
\max_{i \in \ov{S}} \|u_i\|_2 &\leq \theta, \label{robust_cond4} \\
\|\bh\|_2 &\leq \tau \sqrt{s}. \label{robust_cond5}
\end{align}
Let noisy measurements $\by = \AP \bx + \be$ be given with
$\|\be\|_2 \leq \eta$. Then the minimizer $\hat{\bx}$ of
\begin{align*}
\min_{\bz \in \mathcal{H}} \|\bz\|_{2,1} \ \ \text{ s.t. } \ \|\AP \bz
- \by\|_2 \leq \eta
\end{align*}
satisfies
$$
\|\bx - \hat{\bx} \|_2 \leq C_1 \sigma_s(\bx)_1 + (C_2 + C_3
\sqrt{s}) \eta,
$$
where
\begin{align*}
&C_1=\left(1+\frac{\beta}{1-\delta}\right) \frac{2}{1-b}, \ \ C_2=2
\frac{\sqrt{1+\delta}}{1-\delta} +
\left(1+\frac{\beta}{1-\delta}\right)\frac{2\gamma
\sqrt{1+\delta}}{(1-\delta)(1-b)} \\
&C_3=\left(1+\frac{\beta}{1-\delta}\right)\frac{2\tau}{1-b}.
\end{align*}
\end{lemma}

In the remainder of this section, we prove a robust and stable
version of the nonuniform recovery result Theorem~\ref{nonuni_main}
for Bernoulli matrices. We also state the result for the
Gaussian case  but do not prove it since it follows very similarly
to the Bernoulli case.

\begin{theorem}\label{robust_main}
Let $\bx \in \mathcal{H}$ and $S \subset [N]$ with cardinality $s$ be
an index set of $s$ largest $\ell_2$-normed entries of $\bx$. Let $A
\in \R^{m \times N}$ be a Bernoulli matrix and $(W_j)_{j=1}^N$ be
given with parameter $\lambda \in [0,1]$. Assume the measurement
model in \eqref{noisy_model} and let $\hat{\bx}$ be a solution to
$(L1)^\eta$. Provided
\begin{align}\label{robust_main_cond}
m \geq C (1 + \|\Lambda_S\|_\infty) \ln(N)\ln(sk) \ln(\vare^{-1}),
\end{align}
then with probability at least $1-\vare$,
\begin{align}\label{recon_error}
\|\bx -\hat{\bx}\|_2 \leq C_1 \sigma_s(\bx)_1 + C_2 \sqrt{s} \eta.
\end{align}
The constants $C,C_1,C_2 >0$ are universal.
\end{theorem}

The proof is analogous to the one of \cite[Theorem
12.22]{Foucart13}. It invokes Lemma~\ref{inexact_robust} which
gives necessary conditions on the measurement matrices for the
robust and stable nonuniform recovery with them. Since the
conditioning assumption \eqref{robust_cond1} requires normalization
of the matrix, we will work with the matrix $\tAP=\frac{1}{\sqrt{m}}
\AP$. Then observe that the optimization problem $(L1)^\eta$ is
equivalent to
\begin{align*}
\min_{\bz \in \mathcal{H}} \|\bz\|_{2,1} \ \ \text{ s.t. } \ \left\|
\tAP \bz - \frac{1}{\sqrt{m}}\by \right\|_2 \leq \eta.
\end{align*}
\begin{proof}
We follow the golfing scheme as in the proof of
Theorem~\ref{nonuni_main}, see Section~\ref{golf_proof}. In
particular, we make the same choices of the parameters $L,r_n,t_n,
q_n, m_n$ as before. We choose $m_n$ as follows
\begin{align*}
m_1 &\geq c (1+\|\Lambda_S\|_\infty) \ln(N) L \ln(2L \vare^{-1}), \\
m_n &\geq c (1+\|\Lambda_S\|_\infty) \ln(N) \ln(2L \vare^{-1}), \ \
n \geq 2.
\end{align*}
These choices change the number of overall samples $m$ only up to a
constant. Then Conditions \eqref{robust_cond1}, \eqref{robust_cond2},
\eqref{robust_cond3}, \eqref{robust_cond4} are all satisfied for
the normalized matrix $\tAP$ with probability at least $1-\vare$
with appropriate choices of variables $\delta, \beta, \gamma,
\theta$. It remains to verify that the vector $\bh \in \R^{md}$
constructed in Section~\ref{golf_proof} as $\bu=\tAP^* \bh$
satisfies Condition~\eqref{robust_cond5}. For simplicity, assume
without loss of generality that the first $L$ values of $n$ are used
in the construction of the dual vector in \eqref{star2}. Then recall
that
$$
\bu= \sum_{n=1}^L \frac{1}{m_n} (\AP^{(n)})^* (\AP^{(n)})_S
\bw^{(n-1)} = \sum_{n=1}^L \frac{m}{m_n} (\tAP^{(n)})^*
(\tAP^{(n)})_S \bw^{(n-1)}.
$$
Hence, $\bu=\tAP^*\bh$ with $\bh^*=((\bh^{(1)})^*, \ldots,
(\bh^{(L)})^*, 0 , \ldots, 0)$ where
$$
\bh^{(n)} = \frac{m}{m_n} (\tAP^{(n)})_S \bw^{(n-1)} \in \R^{m_n d},
\ \ n=1,\ldots, L'.
$$
Then we have
\begin{align*}
\|\bh\|_2^2 &= \sum_{n=1}^L \|\bh^{(n)}\|_2^2 = \sum_{n=1}^L
\frac{m}{m_n} \left\|\sqrt{\frac{m}{m_n}} (\tAP^{(n)})_S \bw^{(n-1)}
\right\|_2^2 \\
&\sum_{n=1}^L \frac{m}{m_n} \left\|\sqrt{\frac{1}{m_n}}
(\AP^{(n)})_S \bw^{(n-1)} \right\|_2^2.
\end{align*}
We also recall the relation \eqref{star1} of the vectors
$\bw^{(n)}$. This gives, for $n \geq 1$,
\begin{align*}
\left\|\sqrt{\frac{1}{m_n}} (\AP^{(n)})_S \bw^{(n-1)} \right\|_2^2
&= \left\langle \frac{1}{m_n} (\AP^{(n)})_S^* (\AP^{(n)})_S
\bw^{(n-1)} ,
\bw^{(n-1)} \right\rangle \\
&=\left\langle \left(\frac{1}{m_n} (\AP^{(n)})_S^* (\AP^{(n)})_S
-\bP_S \right)\bw^{(n-1)},\bw^{(n-1)} \right\rangle +
\|\bw^{(n-1)}\|_2^2 \\
&=\la \bw^{(n)},\bw^{(n-1)}\ra + \|\bw^{(n-1)}\|_2^2 \leq \|\bw^{(n
)}\|_2^2 \|\bw^{(n-1)}\|_2^2 + \|\bw^{(n-1)}\|_2^2.
\end{align*}
Recall from the assumption \eqref{est1} that $\|\bw^{(n)}\|_2^2 \leq
r_n' \|\bw^{(n-1)}\|_2^2 \leq \|\bw^{(n-1)}\|_2^2$. Then we obtain
\begin{align*}
\left\|\sqrt{\frac{1}{m_n}} (\AP^{(n)})_S \bw^{(n-1)} \right\|_2^2
&\leq 2 \|\bw^{(n-1)}\|_2^2 \leq \|\bw^{(0)}\|_2^2 \prod_{j=1}^{n-1}
(r_j')^2 \\
&= 2 \|\sgn(\bx)_S\|_2^2 \prod_{j=1}^{n-1} (r_j')^2 = 2s
\prod_{j=1}^{n-1} (r_j')^2.
\end{align*}
Assume that $m \leq C(1+\lambda s) \ln(N)\ln(sk) \ln(2\vare^{-1})$
so that $m$ is just large enough to satisfy
\eqref{robust_main_cond}. Recall the definition of $L=\lceil \ln(s)/\ln\ln (N) \rceil +3$.
Then by our choices of $m_n$, we have $\frac{m}{m_n} \leq L$ for $n
\geq 2$ and $\frac{m}{m_1} \leq c$ for some $c > 0$. (If $m$
is much larger, one can rescale $m_n$
proportionally to achieve the same ratio.) This yields
\begin{align*}
\|\bh\|_2^2 & \leq 2s \sum_{n=1}^L \frac{m}{m_n} \prod_{j=1}^{n-1}
(r_j')^2 \leq 2 s \left( \frac{c}{2 \ln(N)} + \sum_{n=2}^L L \prod_{j=1}^{n-1}
\frac{1}{2\ln{N}} \right) \\
&\leq 2 C' s \left(1 + \frac{L}{2 \ln(N)} \frac{1}{[1-1/ (2\ln(N))]}  \right) \leq C'' s,
\end{align*}
where we used the convention $\prod_{j=1}^0 (r_j')^2 =1$. Therefore, all conditions of
Lemma~\ref{inexact_robust} are satisfied for $\bx$ and $\tAP$ with
probability at least $1-\vare$. This completes the proof.
\end{proof}

We now state the result for Gaussian matrices
and skip the proof.

\begin{theorem}\label{robust_gauss}
Let $\bx \in \mathcal{H}$. Let $A \in \R^{m \times N}$ be a Gaussian
matrix and $(W_j)_{j=1}^N$ be given with parameter $\lambda \in
[0,1]$ and $\dim(W_j)=k$ for all $j$. Assume the measurement model
in \eqref{noisy_model} and let $\hat{\bx}$ be a solution to
$(L1)^\eta$. If
\begin{align*}
m \geq \tilde{C} (1+\lambda s) \ln^2 (6 Nk ) \ln^2(\vare^{-1}),
\end{align*}
then with probability at least $1-\vare-N^{-c}$,
$$
\|\bx -\hat{\bx}\|_2 \leq C_1 \sigma_s(\bx)_1 + C_2 \sqrt{s} \eta.
$$
The constants $\tilde{C},C_1,C_2,c >0$ are universal.
\end{theorem}


\section{Numerical Experiments}
In this section, we present numerical experiments in order to highlight
important aspects of the sparse reconstruction in the fusion frame
(FF) setup. The experiments illustrate our theoretical results and
show that when the subspaces are known, one
can significantly improve the recovery of sparse vectors. In all of
our experiments, we use SPGL1 \cite{Berg07,Berg08} to solve the
$\ell_{2,1}$-minimization problems.

\begin{figure}[!t]
\centering
\subfigure[\label{fig1a}]{\includegraphics[width=3.2in,height=2.3in]{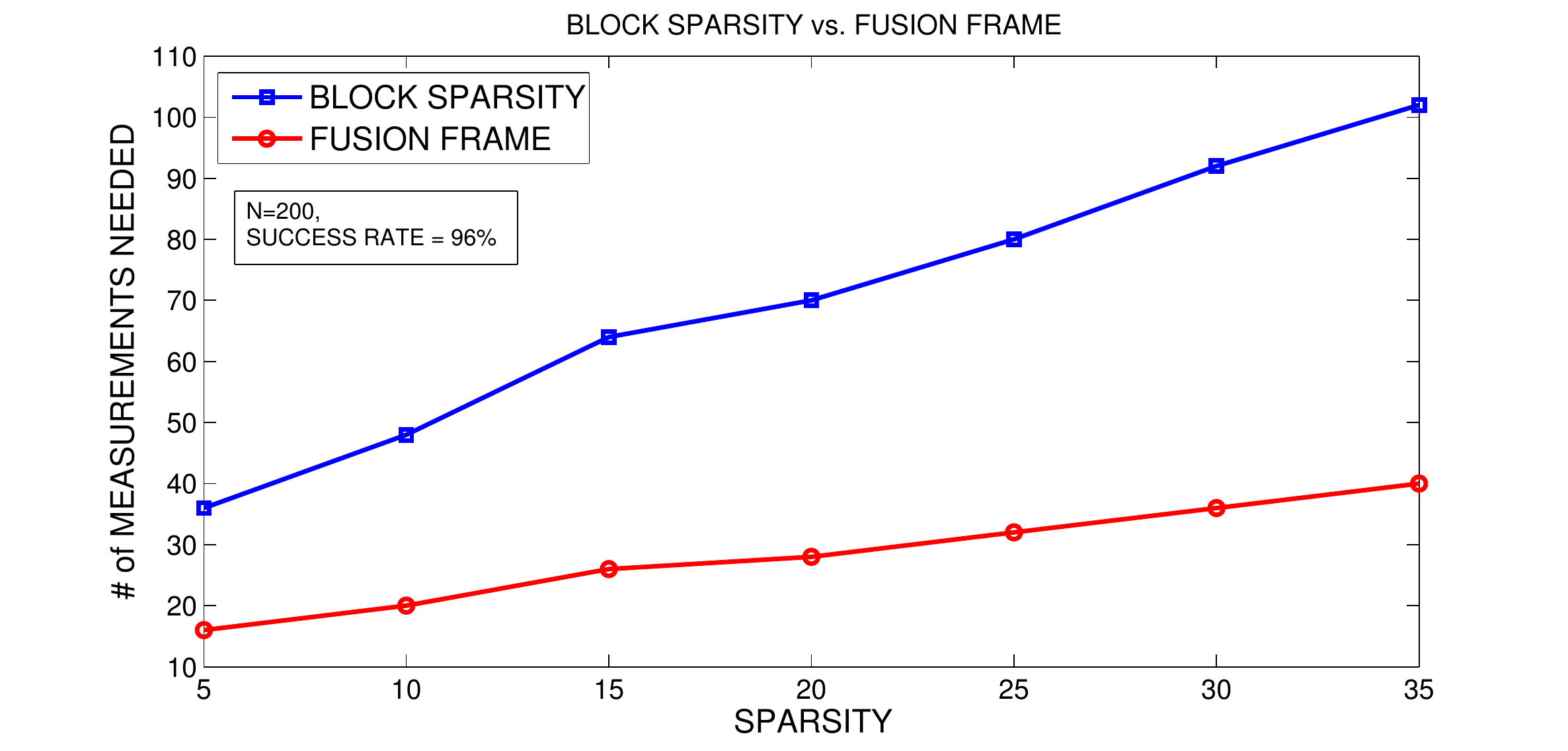}}
\subfigure[\label{fig1b}]{\includegraphics[width=3.2in,height=2.3in]{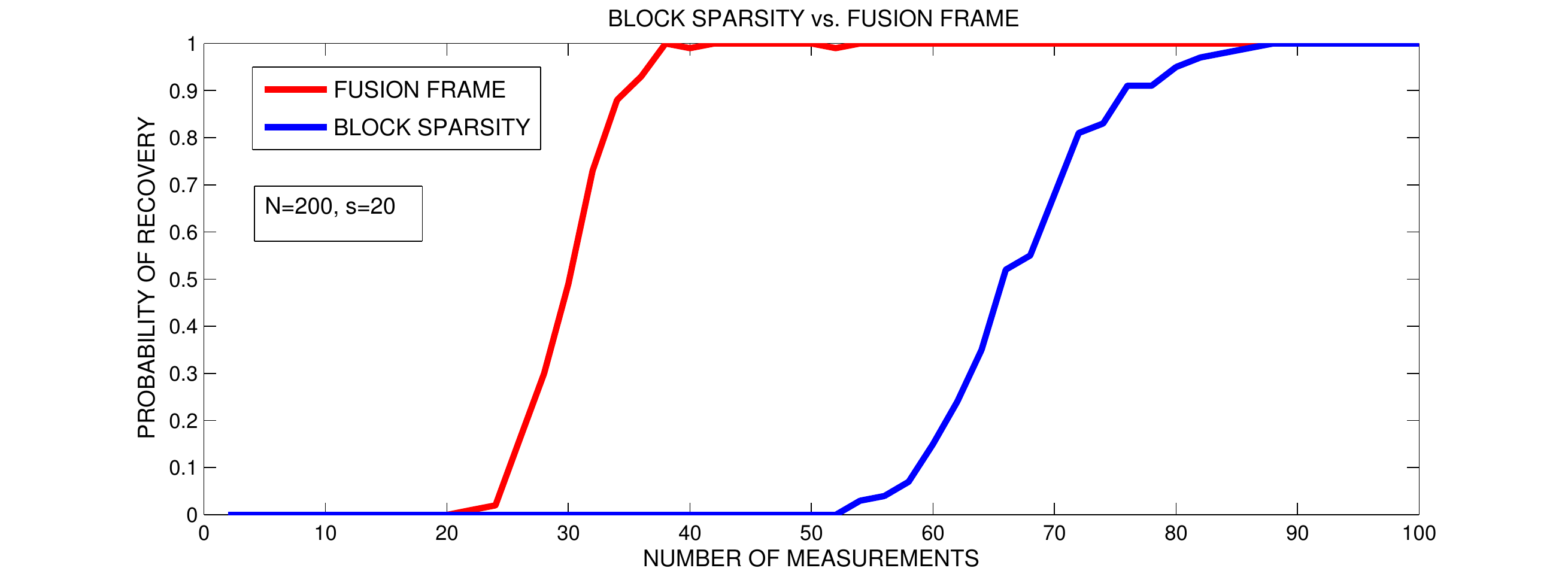}}
\caption{'block' vs. 'fusion frame' sparsity} \label{fig1}
\end{figure}

\paragraph{General setup:}
We generate subspaces randomly, which allows us to generate fusion frames with different values of $\lambda$
and $\Lambda$. Particularly, for $N$ subspaces in $\R^d$ each with
dimension $k$, we generate $N \cdot k$ random vectors from $\mathcal{N}(0,I_d)$ and group them to form the basis for the subspaces. Each
such a random orientation of the subspaces yields a parameter
$\lambda$. In order to obtain a different $\lambda$, it is enough to
vary $d$ or $k$. When $N$ is fixed, $\lambda$ increases with
increasing $k$ and decreasing $d$.

For the measurement matrices, we generate the normalized matrix
$\tilde{A} = \frac{1}{\sqrt{m}} A$ where $A \in \R^{m \times N}$ is
a Gaussian matrix. For a sparsity level $s$, sparse vectors are
generated in the following way: We choose the support set $S$
uniformly at random, then we sample a Gaussian vector in each
subspace in this support set. $N$ is kept fixed throughout the
experiment at hand. In our experiments we work with the parameter
$\|\Lambda_S\|_\infty$ introduced in Section~\ref{incoher}. Since
the random subspaces are not equiangular, this parameter reflects
the linear relation between $m$ and $s$ better than $\lambda$. We
work with the normalized parameter
$$
\lambda_{\eff} = \frac{\|\Lambda_S\|_\infty}{s}.
$$

\paragraph{Exact sparse case:}
In Fig.~\ref{fig1}, we show that the knowledge of the subspaces
improves the recovery. To that end, we fix a fusion frame with
$N=200$ subspaces in $\R^d$ with $\lambda_{\eff} \approx 0.6$. Then
we vary the sparsity level $s$ from $5$ to $35$, and generate an
$s$-sparse vector $\bx$ in the fusion frame. For each $s$, we vary
the number of measurements $m$ and compute empirical recovery rates
via the programs
\begin{align}\label{FF_prog}
\text{(FF)} \ \hat{\bx}= \text{argmin}_{\bx \in \mathcal{H}}
\|\bx\|_{2,1} \ \ s.t. \ \ \AP \bx=\by,
\end{align}
\begin{align}\label{block_prog}
\text{(block)} \ \hat{\bx}= \text{argmin}_{\bx} \|\bx\|_{2,1} \ \
s.t. \ \ {\bf A_I} \bx=\by.
\end{align}
For the whole period, we leave the vector to be recovered fixed. Repeating
this test 100 times with different random $A$ for each choice of
parameters $(s,m,N)$ provides an empirical estimate of the success
probability. In Fig.~\ref{fig1a}, we plot $m$ which yields at least
$96\%$ success rate for each $s$. The difference in two plots is due
to the incoherence of the subspaces, i.e., $\lambda{\eff}$.
Fig.~\ref{fig1b} ($d=3$, $k=1$) shows the transition from the
unsuccessful regime to the successful regime for the sparsity level
$s=20$ for both cases (FF and block sparsity). The transition for
the FF case occurs at a smaller value $m$ which reflects
Fig.~\ref{fig1a} in a different way. As a consequence, the
assumption $\bx \in \mathcal{H}$ in \eqref{FF_prog}
allows us to to recover $\bx$ with far less measurements compared to
\eqref{block_prog} where such a constraint is not used.

Theorem~\ref{nonuni_main} suggests that there is a linear relation
between the number of measurements $m$ and the parameter
$\lambda_{\eff}$. The experiment depicted in Fig.~\ref{fig2} is
designed to reflect this relation. We generate fusion frames with
$N=180$ subspaces with various $\lambda_{\eff}$ which is managed by
changing $d$ and keeping $k=3$ fixed. Then in each fusion frame, a
vector $\bx$ with sparsity $s=25$ is generated and the number of
measurements $m$ that suffices for recovery is determined. The plot
yields an almost linear relation in parallel to the theoretical
result.

\begin{figure}[!t]
\centering
\includegraphics[width=3.2in,height=2.3in]{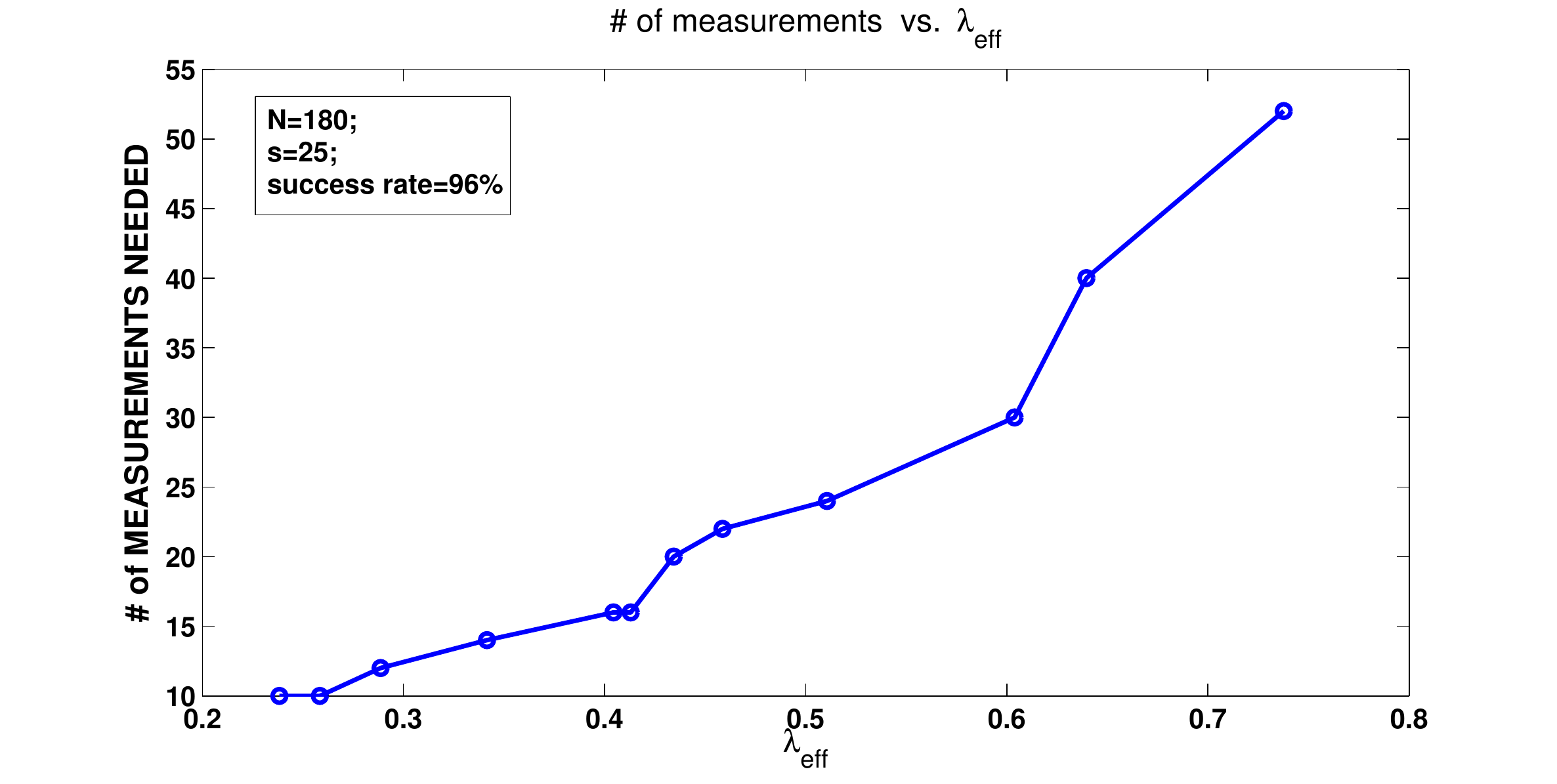}
\caption{$m$ vs. $\lambda_{\eff}$} \label{fig2}
\end{figure}

\paragraph{Stable case:} In this part, we generate scenarios that
allude to the conclusions of Theorems~\ref{robust_main} and
\ref{robust_gauss}. In a fusion frame of $N=200$
subspaces, we generate a signal $\bx$ composed of $\bx_S$, supported
on an index set $S$, and a signal $\bz_{\ov{S}}$ supported on
$\ov{S}$. We then normalize $\bx_S$ and $\bz_{\ov{S}}$ so that
$\|\bx_S\|_{2,1}= \|\bz_{\ov{S}}\|_{2,1}=1$ and produce $\bx = \bx_S +
\theta \bz_{\ov{S}}$ where $\theta \in [0,1]$. Then $\bx$ is our
compressible vector where compressibility is controlled with
$\theta$. For measurement, we choose the normalized Gaussian matrix
$A \in \R^{m\times N}$. We measure $\by = \AP \bx$ and then run the
program $(L1)$  and measure the reconstruction error $\|\bx -
\hat{\bx}\|_2$. We repeat this test $20$ times for a fixed $\bx$
with $\theta=0.12$ in order to obtain an average recovery error for
different values of $m$. Fig.~\ref{fig3a} reports the results of
this experiment performed for different fusion frames with various
values of $\lambda_{\eff}$ and also for the block sparsity case. The
decrease in the reconstruction error with increasing $m$ is natural
even though it is not suggested directly by the theoretical results.
Indeed, one would expect that increasing the number of measurements
would enhance the recovery conditions and yield an improved
reconstruction.

For the noisy case, similarly, we generate noisy observations $\AP
\bx_S + \sigma \be$, of a sparse signal $\bx_S$ where $\|\bx_S\|_2 =
\|\be\|_2 =1$ and $\sigma=0.06$. Here, all entries of the noise
vector $\be$ are chosen i.i.d from the standard Gaussian
distribution and then properly normalized. We then run the robust
$(L1)^\eta$ program and measure the reconstruction error
$\|\bx - \hat{\bx}\|_2$. We plot the average of this error vs. the
number of measurements in Fig~\ref{fig3b} for different values of
$\lambda_{\eff}$.

\begin{figure}[!t]
\centering
\subfigure[\label{fig3a}]{\includegraphics[width=3.2in,height=2.3in]{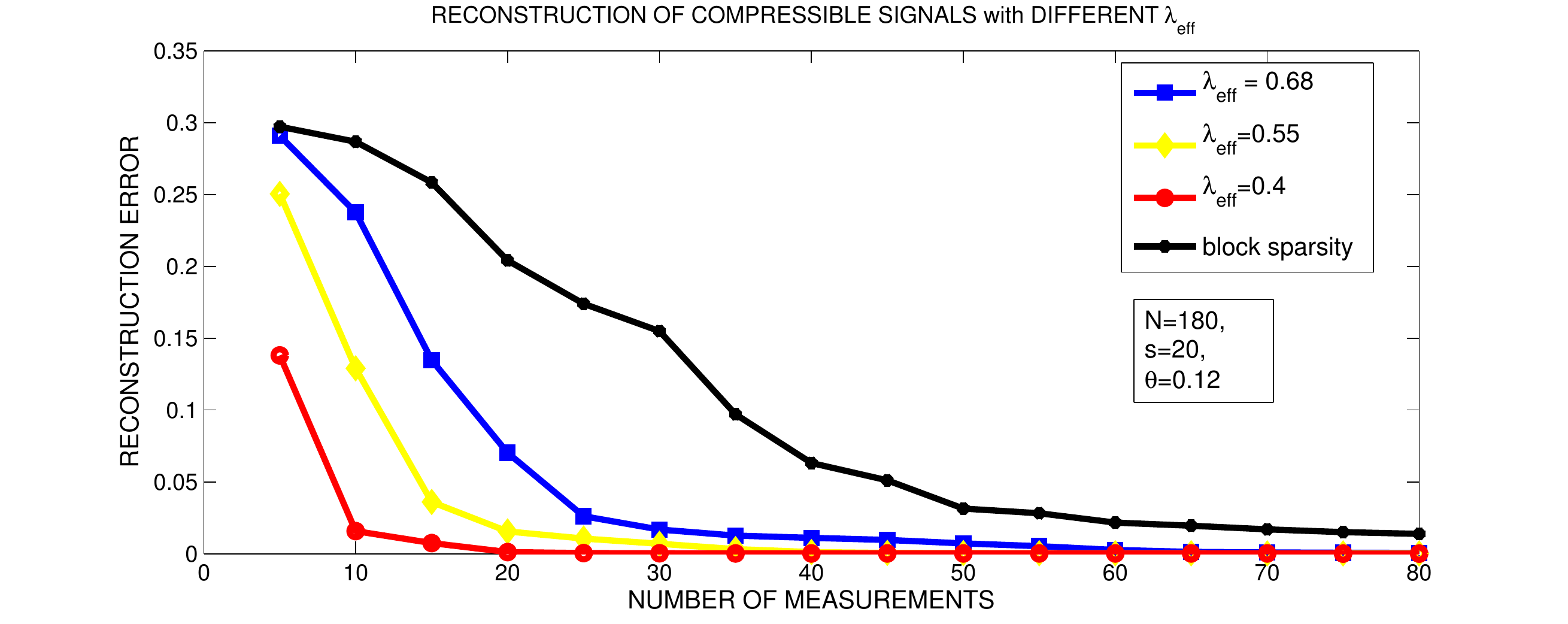}}
\subfigure[\label{fig3b}]{\includegraphics[width=3.2in,height=2.3in]{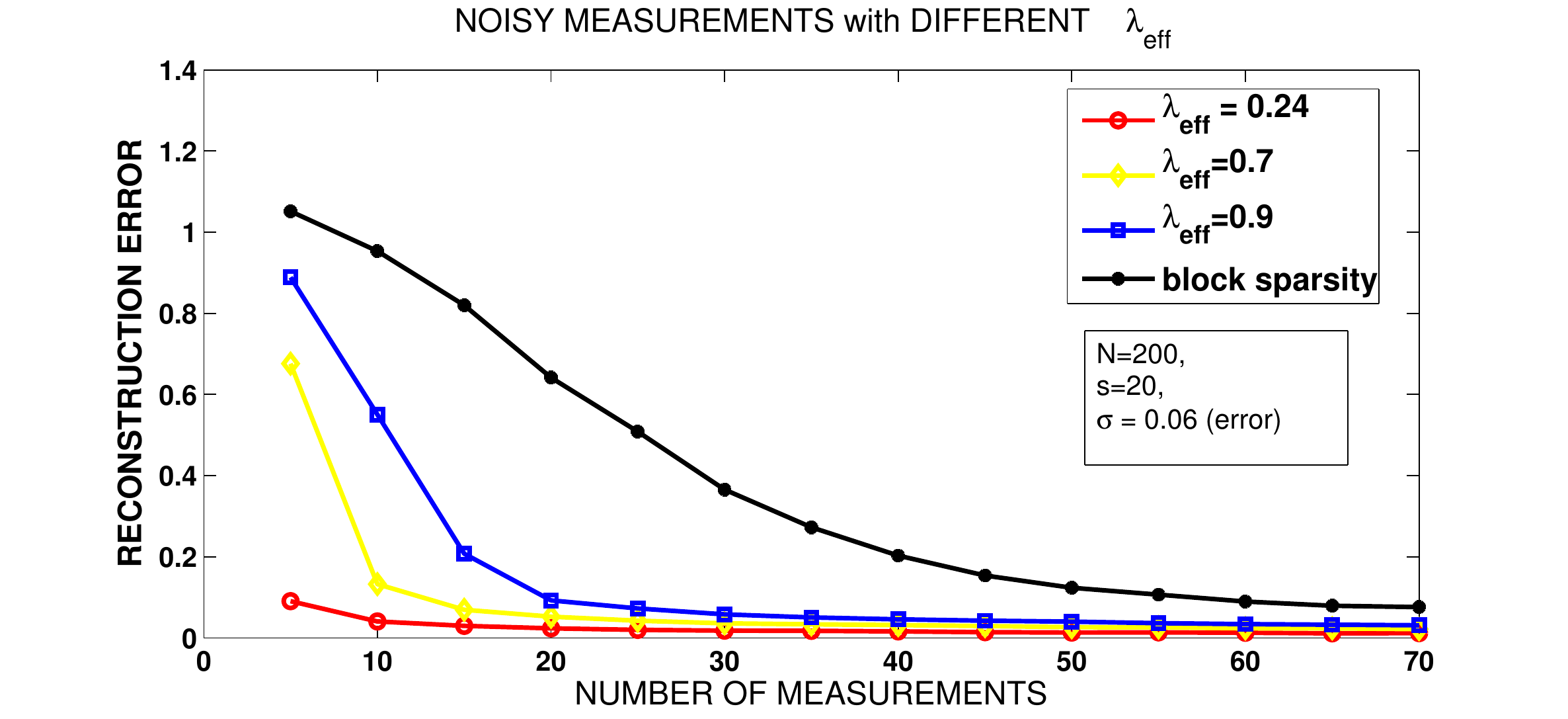}}
\caption{ $\|\hat{\bx}-\bx\|_2$ vs. 'm' }\label{fig3}
\end{figure}

Fig.~\ref{fig4a} depicts the relation between the reconstruction
error and the noise level $\sigma$ for different values of
$\lambda_{\eff}$. In this setup, $N=200$, $s=30$ and $m=50$ are
fixed, and a sparse vector $\bx$ in the fusion frame with specific
value of $\lambda_{\eff}$ is generated. For each value of $\sigma$
we plot the average reconstruction error. Results manifest the
linear relation between $\sigma$ and $\|\bx - \hat{\bx}\|_2$ given
in \eqref{recon_error}. Again, we obtain a better reconstruction quality when
$\lambda_{\eff}$ is smaller.

Finally, we examine the relation between compressibility and the
reconstruction error using a different model than described earlier.
In Fig.~\ref{fig4b}, we plot the results of an experiment in which we
generate signals $\bx$ in a fusion frame with $N=200$, with sorted
values of $\|x_j\|_2$ that decay according to some power law. In
particular, for various values of $ 0 < q <1$, we set $\|x_j\|_2 = c
j^{-1/q}$ such that $\|\bx\|_2=1$. We then measure $\bx$ with
Gaussian matrices $A$ and compute the average reconstruction errors
via $(L1)$ program. Note that the higher the value of $q$, the less
compressible the signal is. The results indicate that reconstruction
of error decreases when the compressibility of the signal increases as declared in
\eqref{recon_error}. We can also see the improvement in the
reconstruction when the subspaces are more incoherent, i.e., they
have smaller $\lambda_{\eff}$.

\begin{figure}[!t]
\centering
\subfigure[\label{fig4a}]{\includegraphics[width=3.2in,height=2.3in]{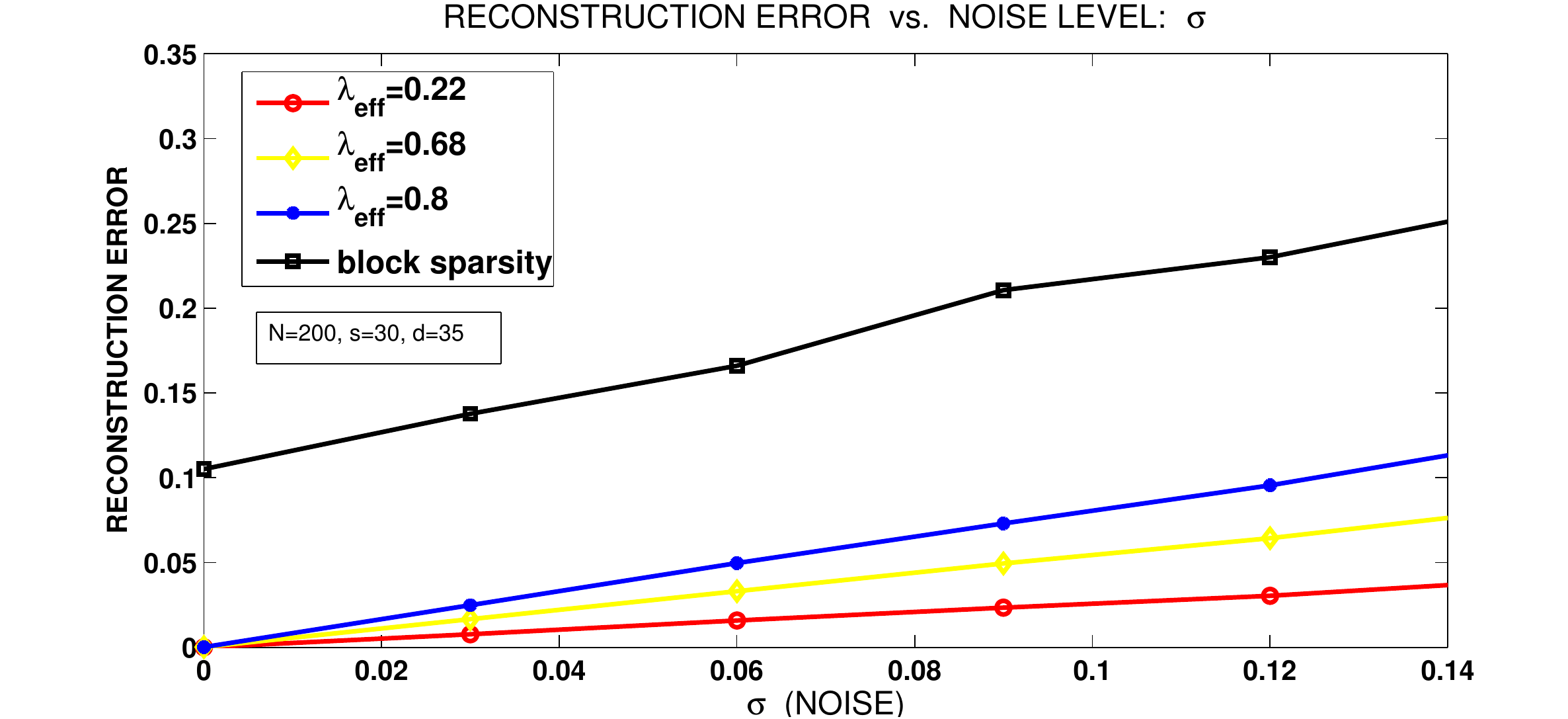}}
\subfigure[\label{fig4b}]{\includegraphics[width=3.2in,height=2.3in]{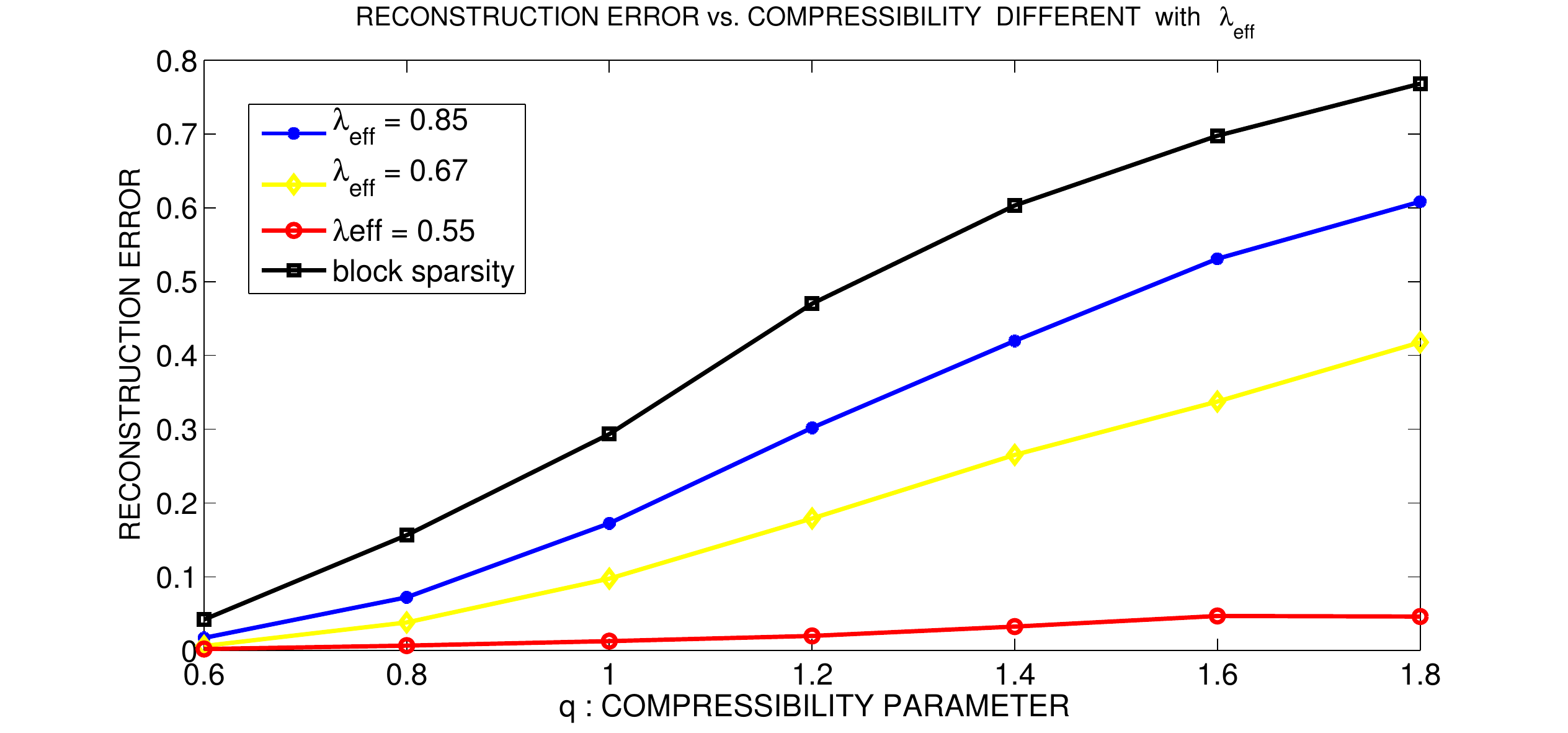}}
\caption{ $\|\hat{\bx}-\bx\|_2$ vs. $\sigma$  \hspace{0.36in} and
  \hspace{0.5in}  $\|\hat{\bx}-\bx\|_2$ vs. $q$}\label{fig4}
\end{figure}




\section{Appendix}


The following theorem is the noncommutative Bernstein inequality
due to Tropp, \cite[Theorem 1.4]{Tropp12}.

\begin{theorem}\label{Tropp}{\bf (Matrix Bernstein inequality)}
Let $\{X_\ell\}_{\ell=1}^M \in \R^{d \times d}$ be a sequence of
independent random self-adjoint matrices. Suppose that $\E X_\ell
=0$ and $\|X_\ell\| \leq K$ a.s. and put
$$
\sigma^2 := \left\| \sum_{\ell=1}^M \E X_\ell^2 \right\|.
$$
Then for all $t \geq 0$,
\begin{align}\label{Tropp_tail}
 \Prob \left( \left\| \sum_{\ell=1}^M X_\ell \right\| \geq t \right) \leq 2 d \exp \left(\frac{-t^2/2}{\sigma^2 + Kt/3}
\right).
\end{align}
\end{theorem}
\vspace{0.005in}
\begin{remark}
One can improve the tail bound \eqref{Tropp_tail}  provided the
$\{X_\ell\}$ are identically distributed and the $\E X_\ell^2$ are not
full rank, say $\rank (\E X_\ell^2) = r < d$. Then \eqref{Tropp_tail} can be replaced by
\begin{align}\label{spec_tail}
 \Prob \left( \left\| \sum_{\ell=1}^M X_\ell \right\| \geq t \right) \leq 2 r \exp \left(\frac{-t^2/2}{\sigma^2 + Kt/3}
\right).
\end{align}
\textit{Sketch of the proof.} We mainly improve
\cite[Corollary~3.7]{Tropp12} under the assumptions above on
$X_\ell$. By \cite[Theorem~3.6]{Tropp12}, for each $\theta > 0$, it
holds that
\begin{align*}
\Prob \left( \lambda_{\max}\left( \sum_{\ell=1}^M X_\ell \right) \geq t
\right) \leq e^{-\theta t} \ \tr \exp\left( \sum_{\ell=1}^M \ln \left(\E e^{\theta
    X_\ell} \right)\right),
\end{align*}
where $\lambda_{\max}$ denotes the largest eigenvalue.
Moreover, \cite[Lemma~6.7]{Tropp12} states that, for $\theta > 0$,
$$
\E e^{\theta X} \preccurlyeq \exp\left( g(\theta) \E X^2 \right)
$$ 
for a self-adjoint, centered random matrix $X$, where $g(\theta)=e^\theta -
\theta -1$. Using this result we obtain
\begin{align*}
\Prob \left( \lambda_{\max}\left( \sum_{\ell=1}^M X_\ell \right) \geq t
\right) &\leq e^{-\theta t} \ \tr \exp\left( g(\theta) \sum_{\ell=1}^M \E
  X_\ell^2 \right) = e^{-\theta t} \ \tr \exp\left( g(\theta) M \E  X_1^2 \right) \\
&\leq e^{-\theta t} r \ \lambda_{\max} \left[ \exp\left( g(\theta) M \E
  X_1^2 \right) \right] = e^{-\theta t} r \ \exp\left( g(\theta) \lambda_{\max} (M \E
  X_1^2 )\right).
\end{align*}
The first equality above uses that $X_\ell$ are identically distributed. 
The second inequality is valid because, for a positive definite matrix $B$ with rank $r$, we have $\tr
B \leq r \lambda_{\max}(B)$ and $\rank (c \ \E
X_\ell^2) =\rank \left(\exp(c \ \E X_\ell^2)\right) = r$, for some $c >
0$. The rest of the proof proceeds in the same way as the proof of \cite[Theorem~1.4]{Tropp12}.
\end{remark}

We also give a rectangular version of the matrix Bernstein
inequality as it appears in \cite[Theorem 1.6]{Tropp12}.

\begin{theorem}\label{Tropp_rect}{\bf (Matrix Bernstein:rectangular)}
Let $\{Z_\ell\} \in \R^{d_1 \times d_2}$ be a finite sequence of
independent random matrices. Suppose that $\E Z_\ell =0$ and
$\|Z_\ell\| \leq K$ a.s. and put
$$
\sigma^2 := \max \left\{ \left\| \sum_\ell \E (Z_\ell Z_\ell^*)
\right\|, \left\| \sum_\ell \E (Z_\ell^* Z_\ell) \right\| \right\}.
$$
Then for all $t \geq 0$,
$$ \Prob \left( \left\| \sum_\ell Z_\ell \right\| \geq t \right) \leq (d_1+d_2) \exp \left(\frac{-t^2/2}{\sigma^2 + Kt/3}
\right).
$$
\end{theorem}

The next lemma is a deviation inequality for sums of independent
random vectors which is a corollary of Bernstein inequalities for
suprema of empirical processes \cite[Corollary 8.44]{Foucart13}. A
similar result can be also found in \cite[Theorem 12]{Gross09}.

\begin{lemma}\label{v_bern}{\bf (Vector Bernstein inequality)} Let $\bY_1, \bY_2, \ldots, \bY_M$ be
independent copies of a random vector $\bY$ on $\R^n$ satisfying $\E
\bY=0$. Assume $\|\bY\|_2 \leq K$. Let
\begin{align*}
Z= \left\| \sum_{\ell=1}^M \bY_\ell \right\|_2, \ \ \E Z^2 = M \E
\|\bY\|_2^2,
\end{align*}
and
\begin{align}\label{bern2}
\sigma^2 = \sup_{\|\bx\|_2 \leq 1} \E |\la x,\bY \ra|^2.
\end{align}
Then, for $t >0$,
\begin{align}\label{bern3}
\Prob (Z \geq \sqrt{\E Z^2} +t) \leq \exp \left(
-\frac{t^2/2}{M\sigma^2 + 2K \sqrt{\E Z^2} + tK/3} \right).
\end{align}
\end{lemma}

\begin{remark}\label{variance}
The so-called weak variance $\sigma^2$ in \eqref{bern2} can be
estimated by
$$ \sigma^2 = \sup_{\|\bx\|_2 \leq 1} \E |\la x,\bY \ra|^2 \leq \E \sup_{\|\bx\|_2 \leq 1} |\la x,\bY
\ra|^2 = \E \|\bY\|_2^2.$$
\end{remark}

\section*{Acknowledgment}
The authors would like to thank the Hausdorff Center for Mathematics
and RWTH Aachen University for support, and acknowledge funding
through the WWTF project SPORTS (MA07-004) and the ERC Starting
Grant StG 258926. 

\bibliographystyle{abbrv}
\bibliography{frame_nonuniform_refer}

\begin{thebibliography}{10}

\bibitem{Ayaz14}
U.~{A}yaz.
\newblock {\em {S}parse Recovery with Fusion Frames}.
\newblock 2014.
\newblock {Ph}.{D}. thesis. Hausdorff Center for Mathematics, University of
  Bonn.

\bibitem{AyazU13}
U.~{A}yaz, S.~{D}irksen, and H.~{R}auhut.
\newblock {U}niform recovery of fusion frame structured sparse signals.
\newblock (Submitted).

\bibitem{Ayaz11}
U.~{A}yaz and H.~{R}auhut.
\newblock {N}onuniform sparse recovery with subgaussian matrices.
\newblock {\em ETNA}, to appear.

\bibitem{Bjorstad91}
P.~E. Bj{\o}rstad and J.~Mandel.
\newblock On the spectra of sums of orthogonal projections with applications to
  parallel computing.
\newblock {\em BIT}, 31(1):76--88, Mar. 1991.

\bibitem{Bodmann07}
B.~G. {B}odmann.
\newblock {O}ptimal linear transmission by loss-sensitive packet encoding.
\newblock {\em {A}ppl. {C}omput. {H}armon. {A}nal.}, 22(3):274--285, 2007.

\bibitem{Boufounos09}
P.~{B}oufounos, G.~{K}utyniok, and H.~{R}auhut.
\newblock {S}parse recovery from combined fusion frame measurements.
\newblock {\em IEEE Trans. Inform. Theory}, 57(6):3864--3876, 2011.

\bibitem{carota06}
E.~{C}and{\`e}s, J.~{R}omberg, and T.~{T}ao.
\newblock {R}obust uncertainty principles: exact signal reconstruction from
  highly incomplete frequency information.
\newblock {\em {I}{E}{E}{E} {T}rans. {I}nf. {T}heory}, 52(2):489--509, 2006.

\bibitem{cata06}
E.~{C}and{\`e}s and T.~{T}ao.
\newblock {N}ear optimal signal recovery from random projections: universal
  encoding strategies?
\newblock {\em {I}{E}{E}{E} {T}rans. {I}nf. {T}heory}, 52(12):5406--5425, 2006.

\bibitem{Casazza13}
P.~{C}asazza and G.~{K}utyniok.
\newblock {\em {F}usion {F}rames}.
\newblock {A}pplied and {N}umerical {H}armonic {A}nalysis. {B}oston, {M}{A}:
  {B}irkh{\"a}user. xvi, 2013.

\bibitem{Casazza04}
P.~G. Casazza and G.~Kutyniok.
\newblock Frames of subspaces.
\newblock {\em in Wavelets, Frames and Operator Theory}, pages 87--113, 2004.

\bibitem{Casazza07}
P.~G. {C}asazza, G.~{K}utyniok, S.~{L}i, and C.~J. {R}ozell.
\newblock {M}odeling sensor networks with fusion frames.
\newblock In {\em {W}avelets {X}ll, {S}pecial {S}ession on
  {F}inite-{D}imensional {F}rames, {T}ime-{F}requency {A}nalysis, and
  {A}pplications}, volume 6701, page~11, 2007.

\bibitem{Chandra12}
V.~Chandrasekaran, B.~Recht, P.~A. Parrilo, and A.~S. Willsky.
\newblock The convex geometry of linear inverse problems.
\newblock {\em Foundations of Computational Mathematics}, 12(6), 2012.

\bibitem{do06-2}
D.~{D}onoho.
\newblock {C}ompressed sensing.
\newblock {\em {I}{E}{E}{E} {T}rans. {I}nform. {T}heory}, 52(4):1289--1306,
  2006.

\bibitem{Eldar08}
Y.~C. Eldar and H.~B{\"o}lcskei.
\newblock Block-sparsity: Coherence and efficient recovery.
\newblock In {\em ICASSP}, pages 2885--2888. IEEE, 2009.

\bibitem{fora08}
M.~{F}ornasier and H.~{R}auhut.
\newblock {R}ecovery algorithms for vector valued data with joint sparsity
  constraints.
\newblock {\em {S}{I}{A}{M} {J}. {N}umer. {A}nal.}, 46(2):577--613, 2008.

\bibitem{Foucart13}
S.~{F}oucart and H.~{R}auhut.
\newblock {\em {A} mathematical introduction to compressive sensing}.
\newblock {A}pplied and {N}umerical {H}armonic {A}nalysis. {B}irkh{\"a}user,
  2013.

\bibitem{Fuchs04}
J.-J. {F}uchs.
\newblock On sparse representations in arbitrary redundant bases.
\newblock {\em IEEE Trans. Inf. Th}, page 1344, 2004.

\bibitem{Gross09}
D.~Gross.
\newblock Recovering low-rank matrices from few coefficients in any basis.
\newblock {\em IEEE Transactions on Information Theory}, 57(3):1548--1566,
  2011.

\bibitem{Kutyniok09}
G.~{K}utyniok, A.~{P}ezeshki, R.~{C}alderbank, and T.~{L}iu.
\newblock {R}obust dimension reduction, fusion frames, and {G}rassmannian
  packings.
\newblock {\em {A}ppl. {C}omput. {H}armon. {A}nal.}, 26(1):64--76, 2009.

\bibitem{Oswald97}
P.~{O}swald.
\newblock {F}rames and space splittings in {H}ilbert spaces.
\newblock Technical report, 1997.

\bibitem{Rao12}
N.~S. Rao, B.~Recht, and R.~D. Nowak.
\newblock Universal measurement bounds for structured sparse signal recovery.
\newblock {\em Jour. of Mach. Learn. Res. - Proc. Track}, 22:942--950, 2012.

\bibitem{Tropp05}
J.~{T}ropp.
\newblock {R}ecovery of short, complex linear combinations via $l_1$
  minimization.
\newblock {\em {I}{E}{E}{E} {T}rans. {I}nf. {T}heory}, 51(4):1568--1570, 2005.

\bibitem{Tropp12}
J.~A. Tropp.
\newblock User-friendly tail bounds for sums of random matrices.
\newblock {\em Found. of Comp. Math.}, 12(4):389--434, 2012.

\bibitem{Berg07}
E.~van~den Berg and M.~P. Friedlander.
\newblock {SPGL1}: A solver for large-scale sparse reconstruction, June 2007.
\newblock http://www.cs.ubc.ca/labs/scl/spgl1.

\bibitem{Berg08}
E.~van~den Berg and M.~P. Friedlander.
\newblock Probing the pareto frontier for basis pursuit solutions.
\newblock {\em SIAM Journ. on Scien. Comp.}, 31(2):890--912, 2008.

\end{thebibliography}

\end{document}